\documentclass{fac}

\usepackage{graphicx}
\usepackage{amssymb}
\usepackage{amsfonts}
\usepackage{amsmath}
\usepackage{dsfont}
\usepackage{MnSymbol}
\usepackage{mathtools}
\usepackage{epsfig}
\usepackage{epstopdf}

\newtheorem{theorem}{Theorem}[section]
\newtheorem{definition}[theorem]{Definition}
\newtheorem{proposition}[theorem]{Proposition}
\newtheorem{lemma}[theorem]{Lemma}

\title[Draft of A Calculus for True Concurrency]
      {A Calculus for True Concurrency}

\author[Yong Wang]
    {Yong Wang\\
     College of Computer Science and Technology,\\
     Faculty of Information Technology,\\
     Beijing University of Technology, Beijing, China\\
     }

\correspond{Yong Wang, Pingleyuan 100, Chaoyang District, Beijing, China.
            e-mail: wangy@bjut.edu.cn}

\pubyear{2017}
\pagerange{\pageref{firstpage}--\pageref{lastpage}}

\begin{document}
\label{firstpage}

\makecorrespond

\maketitle

\begin{abstract}
We design a calculus for true concurrency called CTC, including its syntax and operational semantics. CTC has good properties modulo several kinds of strongly truly concurrent bisimulations and weakly truly concurrent bisimulations, such as monoid laws, static laws, new expansion law for strongly truly concurrent bisimulations, $\tau$ laws for weakly truly concurrent bisimulations, and full congruences for strongly and weakly truly concurrent bisimulations, and also unique solution for recursion.
\end{abstract}


\section{Introduction}\label{int}

Parallelism and concurrency \cite{CM} are the core concepts within computer science. There are mainly two camps in capturing concurrency: the interleaving concurrency and the true concurrency.

The representative of interleaving concurrency is bisimulation/weak bisimulation equivalences. CCS (A Calculus of Communicating Systems) \cite{CCS} \cite{CC} is a calculus based on bisimulation semantics model. CCS has good semantic properties based on the interleaving bisimulation. These properties include monoid laws, static laws, new expansion law for strongly interleaving bisimulation, $\tau$ laws for weakly interleaving bisimulation, and full congruences for strongly and weakly interleaving bisimulations, and also unique solution for recursion.

The other camp of concurrency is true concurrency. The researches on true concurrency are still active. Firstly, there are several truly concurrent bisimulations, the representatives are: pomset bisimulation, step bisimulation, history-preserving (hp-) bisimulation, and especially hereditary history-preserving (hhp-) bisimulation \cite{HHP1} \cite{HHP2}. These truly concurrent bisimulations are studied in different structures \cite{ES1} \cite{ES2} \cite{CM}: Petri nets, event structures, domains, and also a uniform form called TSI (Transition System with Independence) \cite{SFL}. There are also several logics based on different truly concurrent bisimulation equivalences, for example, SFL (Separation Fixpoint Logic) and TFL (Trace Fixpoint Logic) \cite{SFL} are extensions on true concurrency of mu-calculi \cite{MUC} on bisimulation equivalence, and also a logic with reverse modalities \cite {RL1} \cite{RL2} based on the so-called reverse bisimulations with a reverse flavor. Recently, a uniform logic for true concurrency \cite{LTC1} \cite{LTC2} was represented, which used a logical framework to unify several truly concurrent bisimulations, including pomset bisimulation, step bisimulation, hp-bisimulation and hhp-bisimulation.

There are simple comparisons between HM logic and bisimulation, as the uniform logic \cite{LTC1} \cite{LTC2} and truly concurrent bisimulations; the algebraic laws \cite{ALNC}, ACP \cite{ACP} and bisimulation, as the algebraic laws APTC \cite{ATC} and truly concurrent bisimulations; CCS and bisimulation, as truly concurrent bisimulations  and \emph{what}, which is still missing.

In this paper, we design a calculus for true concurrency (CTC) following the way paved by CCS for bisimulation equivalence. This paper is organized as follows. In section \ref{bac}, we introduce some preliminaries, including a brief introduction to CCS, and also preliminaries on true concurrency. We introduce the syntax and operational semantics of CTC in section \ref{sos}, its properties for strongly truly concurrent bisimulations in section \ref{stcb}, its properties for weakly truly concurrent bisimulations in section \ref{wtcb}. In section \ref{app}, we show the applications of CTC by an example called alternating-bit protocol. Finally, in section \ref{con}, we conclude this paper. 

\section{Backgrounds}\label{bac}

\subsection{Process Algebra CCS}

A crucial initial observation that is at the heart of the notion of process algebra is due to Milner, who noticed that concurrent processes have an algebraic structure. CCS \cite{CC} \cite{CCS} is a calculus of concurrent systems. It includes syntax and semantics:

\begin{enumerate}
  \item Its syntax includes actions, process constant, and operators acting between actions, like Prefix, Summation, Composition, Restriction, Relabelling.
  \item Its semantics is based on labeled transition systems, Prefix, Summation, Composition, Restriction, Relabelling have their transition rules. CCS has good semantic properties based on the interleaving bisimulation. These properties include monoid laws, static laws, new expansion law for strongly interleaving bisimulation, $\tau$ laws for weakly interleaving bisimulation, and full congruences for strongly and weakly interleaving bisimulations, and also unique solution for recursion.
\end{enumerate}

CCS can be used widely in verification of computer systems with an interleaving concurrent flavor.

\subsection{True Concurrency}

The related concepts on true concurrency are defined based on the following concepts.

\begin{definition}[Prime event structure with silent event]\label{PES}
Let $\Lambda$ be a fixed set of labels, ranged over $a,b,c,\cdots$ and $\tau$. A ($\Lambda$-labelled) prime event structure with silent event $\tau$ is a tuple $\mathcal{E}=\langle \mathbb{E}, \leq, \sharp, \lambda\rangle$, where $\mathbb{E}$ is a denumerable set of events, including the silent event $\tau$. Let $\hat{\mathbb{E}}=\mathbb{E}\backslash\{\tau\}$, exactly excluding $\tau$, it is obvious that $\hat{\tau^*}=\epsilon$, where $\epsilon$ is the empty event. Let $\lambda:\mathbb{E}\rightarrow\Lambda$ be a labelling function and let $\lambda(\tau)=\tau$. And $\leq$, $\sharp$ are binary relations on $\mathbb{E}$, called causality and conflict respectively, such that:

\begin{enumerate}
  \item $\leq$ is a partial order and $\lceil e \rceil = \{e'\in \mathbb{E}|e'\leq e\}$ is finite for all $e\in \mathbb{E}$. It is easy to see that $e\leq\tau^*\leq e'=e\leq\tau\leq\cdots\leq\tau\leq e'$, then $e\leq e'$.
  \item $\sharp$ is irreflexive, symmetric and hereditary with respect to $\leq$, that is, for all $e,e',e''\in \mathbb{E}$, if $e\sharp e'\leq e''$, then $e\sharp e''$.
\end{enumerate}

Then, the concepts of consistency and concurrency can be drawn from the above definition:

\begin{enumerate}
  \item $e,e'\in \mathbb{E}$ are consistent, denoted as $e\frown e'$, if $\neg(e\sharp e')$. A subset $X\subseteq \mathbb{E}$ is called consistent, if $e\frown e'$ for all $e,e'\in X$.
  \item $e,e'\in \mathbb{E}$ are concurrent, denoted as $e\parallel e'$, if $\neg(e\leq e')$, $\neg(e'\leq e)$, and $\neg(e\sharp e')$.
\end{enumerate}
\end{definition}

The prime event structure without considering silent event $\tau$ is the original one in \cite{ES1} \cite{ES2} \cite{CM}.

\begin{definition}[Configuration]
Let $\mathcal{E}$ be a PES. A (finite) configuration in $\mathcal{E}$ is a (finite) consistent subset of events $C\subseteq \mathcal{E}$, closed with respect to causality (i.e. $\lceil C\rceil=C$). The set of finite configurations of $\mathcal{E}$ is denoted by $\mathcal{C}(\mathcal{E})$. We let $\hat{C}=C\backslash\{\tau\}$.
\end{definition}

Usually, truly concurrent behavioral equivalences are defined by events $e\in\mathcal{E}$ and prime event structure $\mathcal{E}$ (see related concepts in section \ref{STCC} and \ref{WTCC}), in contrast to interleaving behavioral equivalences by actions $a,b\in\mathcal{P}$ and process (graph) $\mathcal{P}$. Indeed, they have correspondences, in \cite{SFL}, models of concurrency, including Petri nets, transition systems and event structures, are unified in a uniform representation -- TSI (Transition System with Independence).

If $x$ is a process, let $C(x)$ denote the corresponding configuration (the already executed part of the process $x$, of course, it is free of conflicts), when $x\xrightarrow{e} x'$, the corresponding configuration $C(x)\xrightarrow{e}C(x')$ with $C(x')=C(x)\cup\{e\}$, where $e$ may be caused by some events in $C(x)$ and concurrent with the other events in $C(x)$, or entirely concurrent with all events in $C(x)$, or entirely caused by all events in $C(x)$. Though the concurrent behavioral equivalences (Definition \ref{PSB}, \ref{WPSB}, \ref{HHPB} and \ref{WHHPB}) are defined based on configurations (pasts of processes), they can also be defined based on processes (futures of configurations), we omit the concrete definitions.

With a little abuse of concepts, in the following of the paper, we will not distinguish actions and events, prime event structures and processes, also concurrent behavior equivalences based on configurations and processes, and use them freely, unless they have specific meanings.

\section{Syntax and Operational Semantics}\label{sos}

We assume an infinite set $\mathcal{N}$ of (action or event) names, and use $a,b,c,\cdots$ to range over $\mathcal{N}$. We denote by $\overline{\mathcal{N}}$ the set of co-names and let $\overline{a},\overline{b},\overline{c},\cdots$ range over $\overline{\mathcal{N}}$. Then we set $\mathcal{L}=\mathcal{N}\cup\overline{\mathcal{N}}$ as the set of labels, and use $l,\overline{l}$ to range over $\mathcal{L}$. We extend complementation to $\mathcal{L}$ such that $\overline{\overline{a}}=a$. Let $\tau$ denote the silent step (internal action or event) and define $Act=\mathcal{L}\cup\{\tau\}$ to be the set of actions, $\alpha,\beta$ range over $Act$. And $K,L$ are used to stand for subsets of $\mathcal{L}$ and $\overline{L}$ is used for the set of complements of labels in $L$. A relabelling function $f$ is a function from $\mathcal{L}$ to $\mathcal{L}$ such that $f(\overline{l})=\overline{f(l)}$. By defining $f(\tau)=\tau$, we extend $f$ to $Act$.

Further, we introduce a set $\mathcal{X}$ of process variables, and a set $\mathcal{K}$ of process constants, and let $X,Y,\cdots$ range over $\mathcal{X}$, and $A,B,\cdots$ range over $\mathcal{K}$, $\widetilde{X}$ is a tuple of distinct process variables, and also $E,F,\cdots$ range over the recursive expressions. We write $\mathcal{P}$ for the set of processes. Sometimes, we use $I,J$ to stand for an indexing set, and we write $E_i:i\in I$ for a family of expressions indexed by $I$. $Id_D$ is the identity function or relation over set $D$.

For each process constant schema $A$, a defining equation of the form

$$A\overset{\text{def}}{=}P$$

is assumed, where $P$ is a process.

\subsection{Syntax}

We use the Prefix $.$ to model the causality relation $\leq$ in true concurrency, the Summation $+$ to model the conflict relation $\sharp$ in true concurrency, and the Composition $\parallel$ to explicitly model concurrent relation in true concurrency. And we follow the conventions of process algebra.

\begin{definition}[Syntax]\label{syntax}
Truly concurrent processes are defined inductively by the following formation rules:

\begin{enumerate}
  \item $A\in\mathcal{P}$;
  \item $\textbf{nil}\in\mathcal{P}$;
  \item if $P\in\mathcal{P}$, then the Prefix $\alpha.P\in\mathcal{P}$, for $\alpha\in Act$;
  \item if $P,Q\in\mathcal{P}$, then the Summation $P+Q\in\mathcal{P}$;
  \item if $P,Q\in\mathcal{P}$, then the Composition $P\parallel Q\in\mathcal{P}$;
  \item if $P\in\mathcal{P}$, then the Prefix $(\alpha_1\parallel\cdots\parallel\alpha_n).P\in\mathcal{P}\quad(n\in I)$, for $\alpha_,\cdots,\alpha_n\in Act$;
  \item if $P\in\mathcal{P}$, then the Restriction $P\setminus L\in\mathcal{P}$ with $L\in\mathcal{L}$;
  \item if $P\in\mathcal{P}$, then the Relabelling $P[f]\in\mathcal{P}$.
\end{enumerate}

The standard BNF grammar of syntax of CTC can be summarized as follows:

$$P::=A\quad|\quad\textbf{nil}\quad|\quad\alpha.P\quad|\quad P+P\quad |\quad P\parallel P\quad |\quad (\alpha_1\parallel\cdots\parallel\alpha_n).P \quad|\quad P\setminus L\quad |\quad P[f].$$
\end{definition}

\subsection{Operational Semantics}

The operational semantics is defined by LTSs (labelled transition systems), and it is detailed by the following definition.

\begin{definition}[Semantics]\label{semantics}
The operational semantics of CTC corresponding to the syntax in Definition \ref{syntax} is defined by a series of transition rules, named $\textbf{Act}$, $\textbf{Sum}$, $\textbf{Com}$, $\textbf{Res}$, $\textbf{Rel}$ and $\textbf{Con}$ indicate that the rules are associated respectively with Prefix, Summation, Composition, Restriction, Relabelling and Constants in Definition \ref{syntax}. They are shown in Table \ref{TRForCTC}.

\begin{center}
    \begin{table}
        \[\textbf{Act}_1\quad \frac{}{\alpha.P\xrightarrow{\alpha}P}\]

        \[\textbf{Sum}_1\quad \frac{P\xrightarrow{\alpha}P'}{P+Q\xrightarrow{\alpha}P'}\]

        \[\textbf{Com}_1\quad \frac{P\xrightarrow{\alpha}P'\quad Q\nrightarrow}{P\parallel Q\xrightarrow{\alpha}P'\parallel Q}\]

        \[\textbf{Com}_2\quad \frac{Q\xrightarrow{\alpha}Q'\quad P\nrightarrow}{P\parallel Q\xrightarrow{\alpha}P\parallel Q'}\]

        \[\textbf{Com}_3\quad \frac{P\xrightarrow{\alpha}P'\quad Q\xrightarrow{\beta}Q'}{P\parallel Q\xrightarrow{\{\alpha,\beta\}}P'\parallel Q'}\quad (\beta\neq\overline{\alpha})\]

        \[\textbf{Com}_4\quad \frac{P\xrightarrow{l}P'\quad Q\xrightarrow{\overline{l}}Q'}{P\parallel Q\xrightarrow{\tau}P'\parallel Q'}\]

        \[\textbf{Act}_2\quad \frac{}{(\alpha_1\parallel\cdots\parallel\alpha_n).P\xrightarrow{\{\alpha_1,\cdots,\alpha_n\}}P}\quad (\alpha_i\neq\overline{\alpha_j}\quad i,j\in\{1,\cdots,n\})\]

        \[\textbf{Sum}_2\quad \frac{P\xrightarrow{\{\alpha_1,\cdots,\alpha_n\}}P'}{P+Q\xrightarrow{\{\alpha_1,\cdots,\alpha_n\}}P'}\]

        \[\textbf{Res}_1\quad \frac{P\xrightarrow{\alpha}P'}{P\setminus L\xrightarrow{\alpha}P'\setminus L}\quad (\alpha,\overline{\alpha}\notin L)\]

        \[\textbf{Res}_2\quad \frac{P\xrightarrow{\{\alpha_1,\cdots,\alpha_n\}}P'}{P\setminus L\xrightarrow{\{\alpha_1,\cdots,\alpha_n\}}P'\setminus L}\quad (\alpha_1,\overline{\alpha_1},\cdots,\alpha_n,\overline{\alpha_n}\notin L)\]

        \[\textbf{Rel}_1\quad \frac{P\xrightarrow{\alpha}P'}{P[f]\xrightarrow{f(\alpha)}P'[f]}\]

        \[\textbf{Rel}_2\quad \frac{P\xrightarrow{\{\alpha_1,\cdots,\alpha_n\}}P'}{P[f]\xrightarrow{\{f(\alpha_1),\cdots,f(\alpha_n)\}}P'[f]}\]

        \[\textbf{Con}_1\quad\frac{P\xrightarrow{\alpha}P'}{A\xrightarrow{\alpha}P'}\quad (A\overset{\text{def}}{=}P)\]

        \[\textbf{Con}_2\quad\frac{P\xrightarrow{\{\alpha_1,\cdots,\alpha_n\}}P'}{A\xrightarrow{\{\alpha_1,\cdots,\alpha_n\}}P'}\quad (A\overset{\text{def}}{=}P)\]

        \caption{Transition rules of CTC}
        \label{TRForCTC}
    \end{table}
\end{center}
\end{definition}

\subsection{Properties of Transitions}

\begin{definition}[Sorts]\label{sorts}
Given the sorts $\mathcal{L}(A)$ and $\mathcal{L}(X)$ of constants and variables, we define $\mathcal{L}(P)$ inductively as follows.

\begin{enumerate}
  \item $\mathcal{L}(l.P)=\{l\}\cup\mathcal{L}(P)$;
  \item $\mathcal{L}((l_1\parallel \cdots\parallel l_n).P)=\{l_1,\cdots,l_n\}\cup\mathcal{L}(P)$;
  \item $\mathcal{L}(\tau.P)=\mathcal{L}(P)$;
  \item $\mathcal{L}(P+Q)=\mathcal{L}(P)\cup\mathcal{L}(Q)$;
  \item $\mathcal{L}(P\parallel Q)=\mathcal{L}(P)\cup\mathcal{L}(Q)$;
  \item $\mathcal{L}(P\setminus L)=\mathcal{L}(P)-(L\cup\overline{L})$;
  \item $\mathcal{L}(P[f])=\{f(l):l\in\mathcal{L}(P)\}$;
  \item for $A\overset{\text{def}}{=}P$, $\mathcal{L}(P)\subseteq\mathcal{L}(A)$.
\end{enumerate}
\end{definition}

Now, we present some properties of the transition rules defined in Table \ref{TRForCTC}.

\begin{proposition}
If $P\xrightarrow{\alpha}P'$, then
\begin{enumerate}
  \item $\alpha\in\mathcal{L}(P)\cup\{\tau\}$;
  \item $\mathcal{L}(P')\subseteq\mathcal{L}(P)$.
\end{enumerate}

If $P\xrightarrow{\{\alpha_1,\cdots,\alpha_n\}}P'$, then
\begin{enumerate}
  \item $\alpha_1,\cdots,\alpha_n\in\mathcal{L}(P)\cup\{\tau\}$;
  \item $\mathcal{L}(P')\subseteq\mathcal{L}(P)$.
\end{enumerate}
\end{proposition}

\begin{proof}
By induction on the inference of $P\xrightarrow{\alpha}P'$ and $P\xrightarrow{\{\alpha_1,\cdots,\alpha_n\}}P'$, there are fourteen cases corresponding to the transition rules named $\textbf{Act}_{1,2}$, $\textbf{Sum}_{1,2}$, $\textbf{Com}_{1,2,3,4}$, $\textbf{Res}_{1,2}$, $\textbf{Rel}_{1,2}$ and $\textbf{Con}_{1,2}$ in Table \ref{TRForCTC}, we just prove the one case $\textbf{Act}_1$ and $\textbf{Act}_2$, and omit the others.

Case $\textbf{Act}_1$: by $\textbf{Act}_1$, with $P\equiv\alpha.P'$. Then by Definition \ref{sorts}, we have (1) $\mathcal{L}(P)=\{\alpha\}\cup\mathcal{L}(P')$ if $\alpha\neq\tau$; (2) $\mathcal{L}(P)=\mathcal{L}(P')$ if $\alpha=\tau$. So, $\alpha\in\mathcal{L}(P)\cup\{\tau\}$, and $\mathcal{L}(P')\subseteq\mathcal{L}(P)$, as desired.

Case $\textbf{Act}_2$: by $\textbf{Act}_2$, with $P\equiv(\alpha_1\parallel\cdots\parallel\alpha_n).P'$. Then by Definition \ref{sorts}, we have (1) $\mathcal{L}(P)=\{\alpha_1,\cdots,\alpha_n\}\cup\mathcal{L}(P')$ if $\alpha_i\neq\tau$ for $i\leq n$; (2) $\mathcal{L}(P)=\mathcal{L}(P')$ if $\alpha_1,\cdots,\alpha_n=\tau$. So, $\alpha_1,\cdots,\alpha_n\in\mathcal{L}(P)\cup\{\tau\}$, and $\mathcal{L}(P')\subseteq\mathcal{L}(P)$, as desired.
\end{proof} 

\section{Strongly Truly Concurrent Bisimulations}\label{stcb}

\subsection{Basic Definitions}\label{STCC}

Firstly, in this subsection, we introduce concepts of (strongly) truly concurrent behavioral bisimulation equivalences, including pomset bisimulation, step bisimulation, history-preserving (hp-)bisimulation and hereditary history-preserving (hhp-)bisimulation.

\begin{definition}[Pomset transitions and step]
Let $\mathcal{E}$ be a PES and let $C\in\mathcal{C}(\mathcal{E})$, and $\emptyset\neq X\subseteq \mathbb{E}$, if $C\cap X=\emptyset$ and $C'=C\cup X\in\mathcal{C}(\mathcal{E})$, then $C\xrightarrow{X} C'$ is called a pomset transition from $C$ to $C'$. When the events in $X$ are pairwise concurrent, we say that $C\xrightarrow{X}C'$ is a step.
\end{definition}

\begin{definition}[Strong pomset, step bisimulation]\label{PSB}
Let $\mathcal{E}_1$, $\mathcal{E}_2$ be PESs. A strong pomset bisimulation is a relation $R\subseteq\mathcal{C}(\mathcal{E}_1)\times\mathcal{C}(\mathcal{E}_2)$, such that if $(C_1,C_2)\in R$, and $C_1\xrightarrow{X_1}C_1'$ then $C_2\xrightarrow{X_2}C_2'$, with $X_1\subseteq \mathbb{E}_1$, $X_2\subseteq \mathbb{E}_2$, $X_1\sim X_2$ and $(C_1',C_2')\in R$, and vice-versa. We say that $\mathcal{E}_1$, $\mathcal{E}_2$ are strong pomset bisimilar, written $\mathcal{E}_1\sim_p\mathcal{E}_2$, if there exists a strong pomset bisimulation $R$, such that $(\emptyset,\emptyset)\in R$. By replacing pomset transitions with steps, we can get the definition of strong step bisimulation. When PESs $\mathcal{E}_1$ and $\mathcal{E}_2$ are strong step bisimilar, we write $\mathcal{E}_1\sim_s\mathcal{E}_2$.
\end{definition}

\begin{definition}[Posetal product]
Given two PESs $\mathcal{E}_1$, $\mathcal{E}_2$, the posetal product of their configurations, denoted $\mathcal{C}(\mathcal{E}_1)\overline{\times}\mathcal{C}(\mathcal{E}_2)$, is defined as

$$\{(C_1,f,C_2)|C_1\in\mathcal{C}(\mathcal{E}_1),C_2\in\mathcal{C}(\mathcal{E}_2),f:C_1\rightarrow C_2 \textrm{ isomorphism}\}.$$

A subset $R\subseteq\mathcal{C}(\mathcal{E}_1)\overline{\times}\mathcal{C}(\mathcal{E}_2)$ is called a posetal relation. We say that $R$ is downward closed when for any $(C_1,f,C_2),(C_1',f',C_2')\in \mathcal{C}(\mathcal{E}_1)\overline{\times}\mathcal{C}(\mathcal{E}_2)$, if $(C_1,f,C_2)\subseteq (C_1',f',C_2')$ pointwise and $(C_1',f',C_2')\in R$, then $(C_1,f,C_2)\in R$.

For $f:X_1\rightarrow X_2$, we define $f[x_1\mapsto x_2]:X_1\cup\{x_1\}\rightarrow X_2\cup\{x_2\}$, $z\in X_1\cup\{x_1\}$,(1)$f[x_1\mapsto x_2](z)=
x_2$,if $z=x_1$;(2)$f[x_1\mapsto x_2](z)=f(z)$, otherwise. Where $X_1\subseteq \mathbb{E}_1$, $X_2\subseteq \mathbb{E}_2$, $x_1\in \mathbb{E}_1$, $x_2\in \mathbb{E}_2$.
\end{definition}

\begin{definition}[Strong (hereditary) history-preserving bisimulation]\label{HHPB}
A strong history-preserving (hp-) bisimulation is a posetal relation $R\subseteq\mathcal{C}(\mathcal{E}_1)\overline{\times}\mathcal{C}(\mathcal{E}_2)$ such that if $(C_1,f,C_2)\in R$, and $C_1\xrightarrow{e_1} C_1'$, then $C_2\xrightarrow{e_2} C_2'$, with $(C_1',f[e_1\mapsto e_2],C_2')\in R$, and vice-versa. $\mathcal{E}_1,\mathcal{E}_2$ are strong history-preserving (hp-)bisimilar and are written $\mathcal{E}_1\sim_{hp}\mathcal{E}_2$ if there exists a strong hp-bisimulation $R$ such that $(\emptyset,\emptyset,\emptyset)\in R$.

A strongly hereditary history-preserving (hhp-)bisimulation is a downward closed strong hp-bisimulation. $\mathcal{E}_1,\mathcal{E}_2$ are strongly hereditary history-preserving (hhp-)bisimilar and are written $\mathcal{E}_1\sim_{hhp}\mathcal{E}_2$.
\end{definition}

\subsection{Laws and Congruence}

Based on the concepts of strongly truly concurrent bisimulation equivalences, we get the following laws.

\begin{proposition}[Monoid laws for strong pomset bisimulation] The monoid laws for strong pomset bisimulation are as follows.

\begin{enumerate}
  \item $P+Q\sim_p Q+P$;
  \item $P+(Q+R)\sim_p (P+Q)+R$;
  \item $P+P\sim_p P$;
  \item $P+\textbf{nil}\sim_p P$.
\end{enumerate}

\end{proposition}

\begin{proof}
\begin{enumerate}
  \item $P+Q\sim_p Q+P$. By the transition rules $\textbf{Sum}_{1,2}$ in Table \ref{TRForCTC}, we get

      $$\frac{P\xrightarrow{p}P'}{P+ Q\xrightarrow{p}P'} (p\subseteq P) \quad \frac{P\xrightarrow{p}P'}{Q+ P\xrightarrow{p}P'}(p\subseteq P)$$

      $$\frac{Q\xrightarrow{q}Q'}{P+ Q\xrightarrow{q}Q'}(q\subseteq Q) \quad \frac{Q\xrightarrow{q}Q'}{Q+ P\xrightarrow{q}Q'}(q\subseteq Q)$$

      Since $P'\sim_p P'$ and $Q'\sim_p Q'$, $P+ Q\sim_p Q+ P$, as desired.
  \item $P+(Q+R)\sim_p (P+Q)+R$. By the transition rules $\textbf{Sum}_{1,2}$ in Table \ref{TRForCTC}, we get

      $$\frac{P\xrightarrow{p}P'}{P+(Q+R)\xrightarrow{p}P'}(p\subseteq P) \quad \frac{P\xrightarrow{p}P'}{(P+Q)+R\xrightarrow{p}P'}(p\subseteq P)$$

      $$\frac{Q\xrightarrow{q}Q'}{P+(Q+R)\xrightarrow{q}Q'}(q\subseteq Q) \quad \frac{Q\xrightarrow{q}Q'}{(P+Q)+R\xrightarrow{q}Q'}(q\subseteq Q)$$

      $$\frac{R\xrightarrow{r}R'}{P+(Q+R)\xrightarrow{r}R'}(r\subseteq R) \quad \frac{R\xrightarrow{r}R'}{(P+Q)+R\xrightarrow{r}R'}(r\subseteq R)$$

      Since $P'\sim_p P'$, $Q'\sim_p Q'$ and $R'\sim_p R'$, $P+(Q+R)\sim_p (P+Q)+R$, as desired.
  \item $P+P\sim_p P$. By the transition rules $\textbf{Sum}_{1,2}$ in Table \ref{TRForCTC}, we get

      $$\frac{P\xrightarrow{p}P'}{P+ P\xrightarrow{p}P'}(p\subseteq P) \quad \frac{P\xrightarrow{p}P'}{P\xrightarrow{p}P'}(p\subseteq P)$$

      Since $P'\sim_p P'$, $P+ P\sim_p P$, as desired.
  \item $P+\textbf{nil}\sim_p P$. By the transition rules $\textbf{Sum}_{1,2}$ in Table \ref{TRForCTC}, we get

      $$\frac{P\xrightarrow{p}P'}{P+ \textbf{nil}\xrightarrow{p}P'}(p\subseteq P) \quad \frac{P\xrightarrow{p}P'}{P\xrightarrow{p}P'}(p\subseteq P)$$

      Since $P'\sim_p P'$, $P+ \textbf{nil}\sim_p P$, as desired.
\end{enumerate}
\end{proof}

\begin{proposition}[Monoid laws for strong step bisimulation] The monoid laws for strong step bisimulation are as follows.
\begin{enumerate}
  \item $P+Q\sim_s Q+P$;
  \item $P+(Q+R)\sim_s (P+Q)+R$;
  \item $P+P\sim_s P$;
  \item $P+\textbf{nil}\sim_s P$.
\end{enumerate}
\end{proposition}

\begin{proof}
\begin{enumerate}
  \item $P+Q\sim_s Q+P$. By the transition rules $\textbf{Sum}_{1,2}$ in Table \ref{TRForCTC}, we get

      $$\frac{P\xrightarrow{p}P'}{P+ Q\xrightarrow{p}P'} (p\subseteq P,\forall\alpha,\beta \in p,\textrm{ are pairwise concurrent})$$

      $$\frac{P\xrightarrow{p}P'}{Q+ P\xrightarrow{p}P'}(p\subseteq P,\forall\alpha,\beta \in p,\textrm{ are pairwise concurrent})$$

      $$\frac{Q\xrightarrow{q}Q'}{P+ Q\xrightarrow{q}Q'}(q\subseteq Q,\forall\alpha,\beta \in q,\textrm{ are pairwise concurrent})$$

      $$\frac{Q\xrightarrow{q}Q'}{Q+ P\xrightarrow{q}Q'}(q\subseteq Q,\forall\alpha,\beta \in q,\textrm{ are pairwise concurrent})$$

      Since $P'\sim_s P'$ and $Q'\sim_s Q'$, $P+ Q\sim_s Q+ P$, as desired.
  \item $P+(Q+R)\sim_s (P+Q)+R$. By the transition rules $\textbf{Sum}_{1,2}$ in Table \ref{TRForCTC}, we get

      $$\frac{P\xrightarrow{p}P'}{P+(Q+R)\xrightarrow{p}P'}(p\subseteq P,\forall\alpha,\beta \in p,\textrm{ are pairwise concurrent})$$

      $$\frac{P\xrightarrow{p}P'}{(P+Q)+R\xrightarrow{p}P'}(p\subseteq P,\forall\alpha,\beta \in p,\textrm{ are pairwise concurrent})$$

      $$\frac{Q\xrightarrow{q}Q'}{P+(Q+R)\xrightarrow{q}Q'}(q\subseteq Q,\forall\alpha,\beta \in q,\textrm{ are pairwise concurrent})$$

      $$\frac{Q\xrightarrow{q}Q'}{(P+Q)+R\xrightarrow{q}Q'}(q\subseteq Q,\forall\alpha,\beta \in q,\textrm{ are pairwise concurrent})$$

      $$\frac{R\xrightarrow{r}R'}{P+(Q+R)\xrightarrow{r}R'}(r\subseteq R,\forall\alpha,\beta \in r,\textrm{ are pairwise concurrent})$$

      $$\frac{R\xrightarrow{r}R'}{(P+Q)+R\xrightarrow{r}R'}(r\subseteq R,\forall\alpha,\beta \in r,\textrm{ are pairwise concurrent})$$

      Since $P'\sim_s P'$, $Q'\sim_s Q'$ and $R'\sim_s R'$, $P+(Q+R)\sim_s (P+Q)+R$, as desired.
  \item $P+P\sim_s P$. By the transition rules $\textbf{Sum}_{1,2}$ in Table \ref{TRForCTC}, we get

      $$\frac{P\xrightarrow{p}P'}{P+ P\xrightarrow{p}P'}(p\subseteq P,\forall\alpha,\beta \in p,\textrm{ are pairwise concurrent})$$

      $$\frac{P\xrightarrow{p}P'}{P\xrightarrow{p}P'}(p\subseteq P,\forall\alpha,\beta \in p,\textrm{ are pairwise concurrent})$$

      Since $P'\sim_s P'$, $P+ P\sim_s P$, as desired.
  \item $P+\textbf{nil}\sim_s P$. By the transition rules $\textbf{Sum}_{1,2}$ in Table \ref{TRForCTC}, we get

      $$\frac{P\xrightarrow{p}P'}{P+ \textbf{nil}\xrightarrow{p}P'}(p\subseteq P,\forall\alpha,\beta \in p,\textrm{ are pairwise concurrent})$$

      $$\frac{P\xrightarrow{p}P'}{P\xrightarrow{p}P'}(p\subseteq P,\forall\alpha,\beta \in p,\textrm{ are pairwise concurrent})$$

      Since $P'\sim_s P'$, $P+ \textbf{nil}\sim_s P$, as desired.
\end{enumerate}
\end{proof}

\begin{proposition}[Monoid laws for strong hp-bisimulation] The monoid laws for strong hp-bisimulation are as follows.
\begin{enumerate}
  \item $P+Q\sim_{hp} Q+P$;
  \item $P+(Q+R)\sim_{hp} (P+Q)+R$;
  \item $P+P\sim_{hp} P$;
  \item $P+\textbf{nil}\sim_{hp} P$.
\end{enumerate}
\end{proposition}

\begin{proof}
\begin{enumerate}
  \item $P+Q\sim_{hp} Q+P$. By the transition rules $\textbf{Sum}_{1,2}$ in Table \ref{TRForCTC}, we get

      $$\frac{P\xrightarrow{\alpha}P'}{P+ Q\xrightarrow{\alpha}P'} \quad \frac{P\xrightarrow{\alpha}P'}{Q+ P\xrightarrow{\alpha}P'}$$

      $$\frac{Q\xrightarrow{\beta}Q'}{P+ Q\xrightarrow{\beta}Q'} \quad \frac{Q\xrightarrow{\beta}Q'}{Q+ P\xrightarrow{\beta}Q'}$$

      Since $(C(P+ Q),f,C(Q+ P))\in\sim_{hp}$, $(C((P+ Q)'),f[\alpha\mapsto \alpha],C((Q+ P)'))\in\sim_{hp}$ and $(C((P+ Q)'),f[\beta\mapsto \beta],C((Q+ P)'))\in\sim_{hp}$, $P+ Q\sim_{hp} Q+ P$, as desired.
  \item $P+(Q+R)\sim_{hp} (P+Q)+R$. By the transition rules $\textbf{Sum}_{1,2}$ in Table \ref{TRForCTC}, we get

      $$\frac{P\xrightarrow{\alpha}P'}{P+(Q+R)\xrightarrow{\alpha}P'} \quad \frac{P\xrightarrow{\alpha}P'}{(P+Q)+R\xrightarrow{\alpha}P'}$$

      $$\frac{Q\xrightarrow{\beta}Q'}{P+(Q+R)\xrightarrow{\beta}Q'} \quad \frac{Q\xrightarrow{\beta}Q'}{(P+Q)+R\xrightarrow{\beta}Q'}$$

      $$\frac{R\xrightarrow{\gamma}R'}{P+(Q+R)\xrightarrow{\gamma}R'} \quad \frac{R\xrightarrow{\gamma}R'}{(P+Q)+R\xrightarrow{\gamma}R'}$$

      Since $(C(P+ (Q+R)),f,C((P+Q)+R))\in\sim_{hp}$, $(C((P+ (Q+R))'),f[\alpha\mapsto \alpha],C((P+Q)+R)'))\in\sim_{hp}$, $(C((P+ (Q+R))'),f[\beta\mapsto \beta],C((P+Q)+R)'))\in\sim_{hp}$ and $(C((P+ (Q+R))'),f[\gamma\mapsto \gamma],C((P+Q)+R)'))\in\sim_{hp}$, $P+(Q+R)\sim_{hp} (P+Q)+R$, as desired.
  \item $P+P\sim_{hp} P$. By the transition rules $\textbf{Sum}_{1,2}$ in Table \ref{TRForCTC}, we get

      $$\frac{P\xrightarrow{\alpha}P'}{P+ P\xrightarrow{\alpha}P'} \quad \frac{P\xrightarrow{\alpha}P'}{P\xrightarrow{\alpha}P'}$$

      Since $(C(P+P),f,C(P))\in\sim_{hp}$, $(C((P+ P)'),f[\alpha\mapsto \alpha],C((P)'))\in\sim_{hp}$, $P+ P\sim_{hp} P$, as desired.
  \item $P+\textbf{nil}\sim_{hp} P$. By the transition rules $\textbf{Sum}_{1,2}$ in Table \ref{TRForCTC}, we get

      $$\frac{P\xrightarrow{\alpha}P'}{P+ \textbf{nil}\xrightarrow{\alpha}P'} \quad \frac{P\xrightarrow{\alpha}P'}{P\xrightarrow{\alpha}P'}$$

      Since $(C(P+\textbf{nil}),f,C(P))\in\sim_{hp}$, $(C((P+ \textbf{nil})'),f[\alpha\mapsto \alpha],C((P)'))\in\sim_{hp}$, $P+ \textbf{nil}\sim_{hp} P$, as desired.
\end{enumerate}
\end{proof}

\begin{proposition}[Monoid laws for strongly hhp-bisimulation] The monoid laws for strongly hhp-bisimulation are as follows.
\begin{enumerate}
  \item $P+Q\sim_{hhp} Q+P$;
  \item $P+(Q+R)\sim_{hhp} (P+Q)+R$;
  \item $P+P\sim_{hhp} P$;
  \item $P+\textbf{nil}\sim_{hhp} P$.
\end{enumerate}
\end{proposition}

\begin{proof}
\begin{enumerate}
  \item $P+Q\sim_{hhp} Q+P$. By the transition rules $\textbf{Sum}_{1,2}$ in Table \ref{TRForCTC}, we get

      $$\frac{P\xrightarrow{\alpha}P'}{P+ Q\xrightarrow{\alpha}P'} \quad \frac{P\xrightarrow{\alpha}P'}{Q+ P\xrightarrow{\alpha}P'}$$

      $$\frac{Q\xrightarrow{\beta}Q'}{P+ Q\xrightarrow{\beta}Q'} \quad \frac{Q\xrightarrow{\beta}Q'}{Q+ P\xrightarrow{\beta}Q'}$$

      Since $(C(P+ Q),f,C(Q+ P))\in\sim_{hhp}$, $(C((P+ Q)'),f[\alpha\mapsto \alpha],C((Q+ P)'))\in\sim_{hhp}$ and $(C((P+ Q)'),f[\beta\mapsto \beta],C((Q+ P)'))\in\sim_{hhp}$, $P+ Q\sim_{hhp} Q+ P$, as desired.
  \item $P+(Q+R)\sim_{hhp} (P+Q)+R$. By the transition rules $\textbf{Sum}_{1,2}$ in Table \ref{TRForCTC}, we get

      $$\frac{P\xrightarrow{\alpha}P'}{P+(Q+R)\xrightarrow{\alpha}P'} \quad \frac{P\xrightarrow{\alpha}P'}{(P+Q)+R\xrightarrow{\alpha}P'}$$

      $$\frac{Q\xrightarrow{\beta}Q'}{P+(Q+R)\xrightarrow{\beta}Q'} \quad \frac{Q\xrightarrow{\beta}Q'}{(P+Q)+R\xrightarrow{\beta}Q'}$$

      $$\frac{R\xrightarrow{\gamma}R'}{P+(Q+R)\xrightarrow{\gamma}R'} \quad \frac{R\xrightarrow{\gamma}R'}{(P+Q)+R\xrightarrow{\gamma}R'}$$

      Since $(C(P+ (Q+R)),f,C((P+Q)+R))\in\sim_{hhp}$, $(C((P+ (Q+R))'),f[\alpha\mapsto \alpha],C((P+Q)+R)'))\in\sim_{hhp}$, $(C((P+ (Q+R))'),f[\beta\mapsto \beta],C((P+Q)+R)'))\in\sim_{hhp}$ and $(C((P+ (Q+R))'),f[\gamma\mapsto \gamma],C((P+Q)+R)'))\in\sim_{hhp}$, $P+(Q+R)\sim_{hhp} (P+Q)+R$, as desired.
  \item $P+P\sim_{hhp} P$. By the transition rules $\textbf{Sum}_{1,2}$ in Table \ref{TRForCTC}, we get

      $$\frac{P\xrightarrow{\alpha}P'}{P+ P\xrightarrow{\alpha}P'} \quad \frac{P\xrightarrow{\alpha}P'}{P\xrightarrow{\alpha}P'}$$

      Since $(C(P+P),f,C(P))\in\sim_{hhp}$, $(C((P+ P)'),f[\alpha\mapsto \alpha],C((P)'))\in\sim_{hhp}$, $P+ P\sim_{hhp} P$, as desired.
  \item $P+\textbf{nil}\sim_{hhp} P$. By the transition rules $\textbf{Sum}_{1,2}$ in Table \ref{TRForCTC}, we get

      $$\frac{P\xrightarrow{\alpha}P'}{P+ \textbf{nil}\xrightarrow{\alpha}P'} \quad \frac{P\xrightarrow{\alpha}P'}{P\xrightarrow{\alpha}P'}$$

      Since $(C(P+\textbf{nil}),f,C(P))\in\sim_{hhp}$, $(C((P+ \textbf{nil})'),f[\alpha\mapsto \alpha],C((P)'))\in\sim_{hhp}$, $P+ \textbf{nil}\sim_{hhp} P$, as desired.
\end{enumerate}
\end{proof}

\begin{proposition}[Static laws for strong step bisimulation] \label{SLSSB}
The static laws for strong step bisimulation are as follows.
\begin{enumerate}
  \item $P\parallel Q\sim_s Q\parallel P$;
  \item $P\parallel(Q\parallel R)\sim_s (P\parallel Q)\parallel R$;
  \item $P\parallel \textbf{nil}\sim_s P$;
  \item $P\setminus L\sim_s P$, if $\mathcal{L}(P)\cap(L\cup\overline{L})=\emptyset$;
  \item $P\setminus K\setminus L\sim_s P\setminus(K\cup L)$;
  \item $P[f]\setminus L\sim_s P\setminus f^{-1}(L)[f]$;
  \item $(P\parallel Q)\setminus L\sim_s P\setminus L\parallel Q\setminus L$, if $\mathcal{L}(P)\cap\overline{\mathcal{L}(Q)}\cap(L\cup\overline{L})=\emptyset$;
  \item $P[Id]\sim_s P$;
  \item $P[f]\sim_s P[f']$, if $f\upharpoonright\mathcal{L}(P)=f'\upharpoonright\mathcal{L}(P)$;
  \item $P[f][f']\sim_s P[f'\circ f]$;
  \item $(P\parallel Q)[f]\sim_s P[f]\parallel Q[f]$, if $f\upharpoonright(L\cup\overline{L})$ is one-to-one, where $L=\mathcal{L}(P)\cup\mathcal{L}(Q)$.
\end{enumerate}
\end{proposition}

\begin{proof}
Though transition rules in Table \ref{TRForCTC} are defined in the flavor of single event, they can be modified into a step (a set of events within which each event is pairwise concurrent), we omit them. If we treat a single event as a step containing just one event, the proof of the static laws does not exist any problem, so we use this way and still use the transition rules in Table \ref{TRForCTC}.

\begin{enumerate}
  \item $P\parallel Q\sim_s Q\parallel P$. By the transition rules $\textbf{Com}_{1,2,3,4}$ in Table \ref{TRForCTC}, we get

      $$\frac{P\xrightarrow{\alpha}P'\quad Q\nrightarrow}{P\parallel Q\xrightarrow{\alpha}P'\parallel Q}
      \quad\frac{P\xrightarrow{\alpha}P'\quad Q\nrightarrow}{Q\parallel P\xrightarrow{\alpha}Q\parallel P'}$$

      $$\frac{Q\xrightarrow{\beta}Q'\quad P\nrightarrow}{P\parallel Q\xrightarrow{\beta}P\parallel Q'}
      \quad\frac{Q\xrightarrow{\beta}Q'\quad P\nrightarrow}{Q\parallel P\xrightarrow{\beta}Q'\parallel P}$$

      $$\frac{P\xrightarrow{\alpha}P'\quad Q\xrightarrow{\beta}Q'}{P\parallel Q\xrightarrow{\{\alpha,\beta\}}P'\parallel Q'}(\beta\neq\overline{\alpha})
      \quad\frac{P\xrightarrow{\alpha}P'\quad Q\xrightarrow{\beta}Q'}{Q\parallel P\xrightarrow{\{\alpha,\beta\}}Q'\parallel P'}(\beta\neq\overline{\alpha})$$

      $$\frac{P\xrightarrow{l}P'\quad Q\xrightarrow{\overline{l}}Q'}{P\parallel Q\xrightarrow{\tau}P'\parallel Q'}
      \quad\frac{P\xrightarrow{l}P'\quad Q\xrightarrow{\overline{l}}Q'}{Q\parallel P\xrightarrow{\tau}Q'\parallel P'}$$

      So, with the assumptions $P'\parallel Q \sim_s Q\parallel P'$, $P\parallel Q' \sim_s Q'\parallel P$ and $P'\parallel Q' \sim_s Q'\parallel P'$, $P\parallel Q\sim_s Q\parallel P$, as desired.
  \item $P\parallel(Q\parallel R)\sim_s (P\parallel Q)\parallel R$. By the transition rules $\textbf{Com}_{1,2,3,4}$ in Table \ref{TRForCTC}, we get

      $$\frac{P\xrightarrow{\alpha}P'\quad Q\nrightarrow\quad R\nrightarrow}{P\parallel (Q\parallel R)\xrightarrow{\alpha}P'\parallel (Q\parallel R)}
      \quad\frac{P\xrightarrow{\alpha}P'\quad Q\nrightarrow\quad R\nrightarrow}{(P\parallel Q)\parallel R\xrightarrow{\alpha}(P'\parallel Q)\parallel R}$$

      $$\frac{Q\xrightarrow{\beta}Q'\quad P\nrightarrow\quad R\nrightarrow}{P\parallel (Q\parallel R)\xrightarrow{\beta}P\parallel (Q'\parallel R)}
      \quad\frac{Q\xrightarrow{\beta}Q'\quad P\nrightarrow\quad R\nrightarrow}{(P\parallel Q)\parallel R\xrightarrow{\beta}(P\parallel Q')\parallel R}$$

      $$\frac{R\xrightarrow{\gamma}R'\quad P\nrightarrow\quad Q\nrightarrow}{P\parallel (Q\parallel R)\xrightarrow{\gamma}P\parallel (Q\parallel R')}
      \quad\frac{R\xrightarrow{\gamma}R'\quad P\nrightarrow\quad Q\nrightarrow}{(P\parallel Q)\parallel R\xrightarrow{\gamma}(P\parallel Q)\parallel R'}$$

      $$\frac{P\xrightarrow{\alpha}P'\quad Q\xrightarrow{\beta}Q'\quad R\nrightarrow}{P\parallel (Q\parallel R)\xrightarrow{\{\alpha,\beta\}}P'\parallel (Q'\parallel R)}(\beta\neq\overline{\alpha})
      \quad\frac{P\xrightarrow{\alpha}P'\quad Q\xrightarrow{\beta}Q'\quad R\nrightarrow}{(P\parallel Q)\parallel R\xrightarrow{\{\alpha,\beta\}}(P'\parallel Q')\parallel R}(\beta\neq\overline{\alpha})$$

      $$\frac{P\xrightarrow{\alpha}P'\quad R\xrightarrow{\gamma}R'\quad Q\nrightarrow}{P\parallel (Q\parallel R)\xrightarrow{\{\alpha,\gamma\}}P'\parallel (Q\parallel R')}(\gamma\neq\overline{\alpha})
      \quad\frac{P\xrightarrow{\alpha}P'\quad R\xrightarrow{\gamma}R'\quad Q\nrightarrow}{(P\parallel Q)\parallel R\xrightarrow{\{\alpha,\gamma\}}(P'\parallel Q)\parallel R]}(\gamma\neq\overline{\alpha})$$

      $$\frac{Q\xrightarrow{\beta}P'\quad R\xrightarrow{\gamma}R'\quad P\nrightarrow}{P\parallel (Q\parallel R)\xrightarrow{\{\beta,\gamma\}}P\parallel (Q'\parallel R')}(\gamma\neq\overline{\beta})
      \quad\frac{Q\xrightarrow{\beta}Q'\quad R\xrightarrow{\gamma}R'\quad P\nrightarrow}{(P\parallel Q)\parallel R\xrightarrow{\{\beta,\gamma\}}(P\parallel Q')\parallel R'}(\gamma\neq\overline{\beta})$$

      $$\frac{P\xrightarrow{\alpha}P'\quad Q\xrightarrow{\beta}Q'\quad R\xrightarrow{\gamma}R'}{P\parallel (Q\parallel R)\xrightarrow{\{\alpha,\beta,\gamma\}}P'\parallel (Q'\parallel R')}(\beta\neq\overline{\alpha},\gamma\neq\overline{\alpha},\gamma\neq\overline{\beta})
      \quad\frac{P\xrightarrow{\alpha}P'\quad Q\xrightarrow{\beta}Q'\quad R\xrightarrow{\gamma}R'}{(P\parallel Q)\parallel R\xrightarrow{\{\alpha,\beta,\gamma\}}(P'\parallel Q')\parallel R'}(\beta\neq\overline{\alpha},\gamma\neq\overline{\alpha},\gamma\neq\overline{\beta})$$

      $$\frac{P\xrightarrow{l}P'\quad Q\xrightarrow{\overline{l}}Q'\quad R\nrightarrow}{P\parallel (Q\parallel R)\xrightarrow{\tau}P'\parallel (Q'\parallel R)}
      \quad\frac{P\xrightarrow{l}P'\quad Q\xrightarrow{\overline{l}}Q'\quad R\nrightarrow}{(P\parallel Q)\parallel R\xrightarrow{\tau}(P'\parallel Q')\parallel R}$$

      $$\frac{P\xrightarrow{l}P'\quad R\xrightarrow{\overline{l}}R'\quad Q\nrightarrow}{P\parallel (Q\parallel R)\xrightarrow{\tau}P'\parallel (Q\parallel R')}
      \quad\frac{P\xrightarrow{l}P'\quad R\xrightarrow{\overline{l}}R'\quad Q\nrightarrow}{(P\parallel Q)\parallel R\xrightarrow{\tau}(P'\parallel Q)\parallel R]}$$

      $$\frac{Q\xrightarrow{l}P'\quad R\xrightarrow{\overline{l}}R'\quad P\nrightarrow}{P\parallel (Q\parallel R)\xrightarrow{\tau}P\parallel (Q'\parallel R')}
      \quad\frac{Q\xrightarrow{l}Q'\quad R\xrightarrow{\overline{l}}R'\quad P\nrightarrow}{(P\parallel Q)\parallel R\xrightarrow{\tau}(P\parallel Q')\parallel R'}$$

      $$\frac{P\xrightarrow{l}P'\quad Q\xrightarrow{\overline{l}}Q'\quad R\xrightarrow{\gamma}R'}{P\parallel (Q\parallel R)\xrightarrow{\tau,\gamma}P'\parallel (Q'\parallel R')}
      \quad\frac{P\xrightarrow{l}P'\quad Q\xrightarrow{\overline{l}}Q'\quad R\xrightarrow{\gamma}R'}{(P\parallel Q)\parallel R\xrightarrow{\tau,\gamma}(P'\parallel Q')\parallel R'}$$

      $$\frac{P\xrightarrow{l}P'\quad R\xrightarrow{\overline{l}}R'\quad Q\xrightarrow{\beta}Q'}{P\parallel (Q\parallel R)\xrightarrow{\tau,\beta}P'\parallel (Q'\parallel R')}
      \quad\frac{P\xrightarrow{l}P'\quad R\xrightarrow{\overline{l}}R'\quad Q\xrightarrow{\beta}Q'}{(P\parallel Q)\parallel R\xrightarrow{\tau,\beta}(P'\parallel Q')\parallel R]}$$

      $$\frac{Q\xrightarrow{l}Q'\quad R\xrightarrow{\overline{l}}R'\quad P\xrightarrow{\alpha}P'}{P\parallel (Q\parallel R)\xrightarrow{\tau,\alpha}P'\parallel (Q'\parallel R')}
      \quad\frac{Q\xrightarrow{l}Q'\quad R\xrightarrow{\overline{l}}R'\quad P\xrightarrow{\alpha}P'}{(P\parallel Q)\parallel R\xrightarrow{\tau,\alpha}(P'\parallel Q')\parallel R'}$$

      So, with the assumptions $P'\parallel (Q\parallel R) \sim_s (P'\parallel Q)\parallel R$, $P\parallel (Q'\parallel R) \sim_s (P\parallel Q')\parallel R$, $P\parallel (Q\parallel R') \sim_s (P\parallel Q)\parallel R'$, $P'\parallel (Q'\parallel R) \sim_s (P'\parallel Q')\parallel R$, $P'\parallel (Q\parallel R') \sim_s (P'\parallel Q)\parallel R'$, $P\parallel (Q'\parallel R') \sim_s (P\parallel Q')\parallel R'$ and $P'\parallel (Q'\parallel R') \sim_s (P'\parallel Q')\parallel R'$, $P\parallel (Q\parallel R) \sim_s (P\parallel Q)\parallel R$, as desired.
  \item $P\parallel \textbf{nil}\sim_s P$. By the transition rules $\textbf{Com}_{1,2,3,4}$ in Table \ref{TRForCTC}, we get

      $$\frac{P\xrightarrow{\alpha}P'}{P\parallel \textbf{nil}\xrightarrow{\alpha}P'} \quad \frac{P\xrightarrow{\alpha}P'}{P\xrightarrow{\alpha}P'}$$

      Since $P'\sim_s P'$, $P\parallel \textbf{nil}\sim_s P$, as desired.
  \item $P\setminus L\sim_s P$, if $\mathcal{L}(P)\cap(L\cup\overline{L})=\emptyset$. By the transition rules $\textbf{Res}_{1,2}$ in Table \ref{TRForCTC}, we get

      $$\frac{P\xrightarrow{\alpha}P'}{P\setminus L\xrightarrow{\alpha}P'\setminus L} (\mathcal{L}(P)\cap(L\cup\overline{L})=\emptyset)\quad \frac{P\xrightarrow{\alpha}P'}{P\xrightarrow{\alpha}P'}$$

      Since $P'\sim_s P'$, and with the assumption $P'\setminus L\sim_s P'$, $P\setminus L\sim_s P$, if $\mathcal{L}(P)\cap(L\cup\overline{L})=\emptyset$, as desired.
  \item $P\setminus K\setminus L\sim_s P\setminus(K\cup L)$. By the transition rules $\textbf{Res}_{1,2}$ in Table \ref{TRForCTC}, we get

      $$\frac{P\xrightarrow{\alpha}P'}{P\setminus K\setminus L\xrightarrow{\alpha}P'\setminus K\setminus L} \quad \frac{P\xrightarrow{\alpha}P'}{P\setminus (K\cup L)\xrightarrow{\alpha}P'\setminus (K\cup L)}$$

      Since $P'\sim_s P'$, and with the assumption $P'\setminus K\setminus L\sim_s P'\setminus(K\cup L)$, $P\setminus K\setminus L\sim_s P\setminus(K\cup L)$, as desired.
  \item $P[f]\setminus L\sim_s P\setminus f^{-1}(L)[f]$. By the transition rules $\textbf{Res}_{1,2}$ and $\textbf{Rel}_{1,2}$ in Table \ref{TRForCTC}, we get

      $$\frac{P\xrightarrow{\alpha}P'}{P[f]\setminus L\xrightarrow{f(\alpha)}P'[f]\setminus L}\quad \frac{P\xrightarrow{\alpha}P'}{P\setminus f^{-1}(L)[f]\xrightarrow{f(\alpha)}P'\setminus f^{-1}(L)[f]}$$

      So, with the assumption $P'[f]\setminus L\sim_s P'\setminus f^{-1}(L)[f]$, $P[f]\setminus L\sim_s P\setminus f^{-1}(L)[f]$, as desired.
  \item $(P\parallel Q)\setminus L\sim_s P\setminus L\parallel Q\setminus L$, if $\mathcal{L}(P)\cap\overline{\mathcal{L}(Q)}\cap(L\cup\overline{L})=\emptyset$. By the transition rules $\textbf{Com}_{1,2,3,4}$ and $\textbf{Res}_{1,2}$ in Table \ref{TRForCTC}, we get

      $$\frac{P\xrightarrow{\alpha}P'\quad Q\nrightarrow}{(P\parallel Q)\setminus L\xrightarrow{\alpha}(P'\parallel Q)\setminus L}(\mathcal{L}(P)\cap\overline{\mathcal{L}(Q)}\cap(L\cup\overline{L})=\emptyset)$$
      $$\frac{P\xrightarrow{\alpha}P'\quad Q\nrightarrow}{P\setminus L\parallel Q\setminus L\xrightarrow{\alpha}P'\setminus L\parallel Q\setminus L}(\mathcal{L}(P)\cap\overline{\mathcal{L}(Q)}\cap(L\cup\overline{L})=\emptyset)$$

      $$\frac{Q\xrightarrow{\beta}Q'\quad P\nrightarrow}{(P\parallel Q)\setminus L\xrightarrow{\beta}(P\parallel Q')\setminus L}(\mathcal{L}(P)\cap\overline{\mathcal{L}(Q)}\cap(L\cup\overline{L})=\emptyset)$$
      $$\frac{Q\xrightarrow{\beta}Q'\quad P\nrightarrow}{P\setminus L\parallel Q\setminus L\xrightarrow{\beta}P\setminus L\parallel Q'\setminus L}(\mathcal{L}(P)\cap\overline{\mathcal{L}(Q)}\cap(L\cup\overline{L})=\emptyset)$$

      $$\frac{P\xrightarrow{\alpha}P'\quad Q\xrightarrow{\beta}Q'}{(P\parallel Q)\setminus L\xrightarrow{\{\alpha,\beta\}}(P'\parallel Q')\setminus L}(\mathcal{L}(P)\cap\overline{\mathcal{L}(Q)}\cap(L\cup\overline{L})=\emptyset)$$
      $$\frac{P\xrightarrow{\alpha}P'\quad Q\xrightarrow{\beta}Q'}{P\setminus L\parallel Q\setminus L\xrightarrow{\{\alpha,\beta\}}(P'\parallel Q')\setminus L}(\mathcal{L}(P)\cap\overline{\mathcal{L}(Q)}\cap(L\cup\overline{L})=\emptyset)$$

      $$\frac{P\xrightarrow{l}P'\quad Q\xrightarrow{\overline{l}}Q'}{(P\parallel Q)\setminus L\xrightarrow{\tau}(P'\parallel Q')\setminus L}(\mathcal{L}(P)\cap\overline{\mathcal{L}(Q)}\cap(L\cup\overline{L})=\emptyset)$$
      $$\frac{P\xrightarrow{l}P'\quad Q\xrightarrow{\overline{l}}Q'}{(P\setminus L\parallel Q\setminus L\xrightarrow{\tau}P'\setminus L\parallel Q'\setminus L}(\mathcal{L}(P)\cap\overline{\mathcal{L}(Q)}\cap(L\cup\overline{L})=\emptyset)$$

      Since $(P'\parallel Q)\setminus L\sim_s P'\setminus L\parallel Q\setminus L$, $(P\parallel Q')\setminus L\sim_s P\setminus L\parallel Q'\setminus L$ and $(P'\parallel Q')\setminus L\sim_s P'\setminus L\parallel Q'\setminus L$, $(P\parallel Q)\setminus L\sim_s P\setminus L\parallel Q\setminus L$, if $\mathcal{L}(P)\cap\overline{\mathcal{L}(Q)}\cap(L\cup\overline{L})=\emptyset$, as desired.
  \item $P[Id]\sim_s P$. By the transition rules $\textbf{Rel}_{1,2}$ in Table \ref{TRForCTC}, we get

      $$\frac{P\xrightarrow{\alpha}P'}{P[Id]\xrightarrow{Id(\alpha)}P'[Id]}\quad \frac{P\xrightarrow{\alpha}P'}{P\xrightarrow{\alpha}P'}$$

      So, with the assumption $P'[Id]\sim_s P'$ and $Id(\alpha)=\alpha$, $P[Id]\sim_s P$, as desired.
  \item $P[f]\sim_s P[f']$, if $f\upharpoonright\mathcal{L}(P)=f'\upharpoonright\mathcal{L}(P)$. By the transition rules $\textbf{Rel}_{1,2}$ in Table \ref{TRForCTC}, we get

      $$\frac{P\xrightarrow{\alpha}P'}{P[f]\xrightarrow{f(\alpha)}P'[f]}\quad \frac{P\xrightarrow{\alpha}P'}{P[f']\xrightarrow{f'(\alpha)}P'[f']}$$

      So, with the assumption $P'[f]\sim_s P'[f']$ and $f(\alpha)=f'(\alpha)$, if $f\upharpoonright\mathcal{L}(P)=f'\upharpoonright\mathcal{L}(P)$, $P[f]\sim_s P[f']$, as desired.
  \item $P[f][f']\sim_s P[f'\circ f]$. By the transition rules $\textbf{Rel}_{1,2}$ in Table \ref{TRForCTC}, we get

      $$\frac{P\xrightarrow{\alpha}P'}{P[f][f']\xrightarrow{f'(f(\alpha))}P'[f][f']}\quad \frac{P\xrightarrow{\alpha}P'}{P[f'\circ f]\xrightarrow{f'(f(\alpha))}P'[f'\circ f]}$$

      So, with the assumption $P'[f][f']\sim_s P'[f'\circ f]$, $P[f][f']\sim_s P[f'\circ f]$, as desired.
  \item $(P\parallel Q)[f]\sim_s P[f]\parallel Q[f]$, if $f\upharpoonright(L\cup\overline{L})$ is one-to-one, where $L=\mathcal{L}(P)\cup\mathcal{L}(Q)$. By the transition rules $\textbf{Com}_{1,2,3,4}$ and $\textbf{Rel}_{1,2}$ in Table \ref{TRForCTC}, we get

      $$\frac{P\xrightarrow{\alpha}P'\quad Q\nrightarrow}{(P\parallel Q)[f]\xrightarrow{f(\alpha)}(P'\parallel Q)[f]}(\textrm{if } f\upharpoonright(L\cup\overline{L}) \textrm{ is one-to-one, where }L=\mathcal{L}(P)\cup\mathcal{L}(Q))$$
      $$\frac{P\xrightarrow{\alpha}P'\quad Q\nrightarrow}{P[f]\parallel Q[f]\xrightarrow{f(\alpha)}P'[f]\parallel Q[f]}(\textrm{if } f\upharpoonright(L\cup\overline{L}) \textrm{ is one-to-one, where }L=\mathcal{L}(P)\cup\mathcal{L}(Q))$$

      $$\frac{Q\xrightarrow{\beta}Q'\quad P\nrightarrow}{(P\parallel Q)[f]\xrightarrow{f(\beta)}(P\parallel Q')[f]}(\textrm{if } f\upharpoonright(L\cup\overline{L}) \textrm{ is one-to-one, where }L=\mathcal{L}(P)\cup\mathcal{L}(Q))$$
      $$\frac{Q\xrightarrow{\beta}Q'\quad P\nrightarrow}{P[f]\parallel Q[f]\xrightarrow{f(\beta)}P[f]\parallel Q'[f]}(\textrm{if } f\upharpoonright(L\cup\overline{L}) \textrm{ is one-to-one, where }L=\mathcal{L}(P)\cup\mathcal{L}(Q))$$

      $$\frac{P\xrightarrow{\alpha}P'\quad Q\xrightarrow{\beta}Q'}{(P\parallel Q)[f]\xrightarrow{\{f(\alpha),f(\beta)\}}(P'\parallel Q')[f]}(\textrm{if } f\upharpoonright(L\cup\overline{L}) \textrm{ is one-to-one, where }L=\mathcal{L}(P)\cup\mathcal{L}(Q))$$
      $$\frac{P\xrightarrow{\alpha}P'\quad Q\xrightarrow{\beta}Q'}{P[f]\parallel Q[f]\xrightarrow{\{f(\alpha),f(\beta)\}}P'[f]\parallel Q'[f]}(\textrm{if } f\upharpoonright(L\cup\overline{L}) \textrm{ is one-to-one, where }L=\mathcal{L}(P)\cup\mathcal{L}(Q))$$

      $$\frac{P\xrightarrow{l}P'\quad Q\xrightarrow{\overline{l}}Q'}{(P\parallel Q)[f]\xrightarrow{\tau}(P'\parallel Q')[f]}(\textrm{if } f\upharpoonright(L\cup\overline{L}) \textrm{ is one-to-one, where }L=\mathcal{L}(P)\cup\mathcal{L}(Q))$$
      $$\frac{P\xrightarrow{l}P'\quad Q\xrightarrow{\overline{l}}Q'}{(P[f]\parallel Q[f]\xrightarrow{\tau}P'[f]\parallel Q'[f]}(\textrm{if } f\upharpoonright(L\cup\overline{L}) \textrm{ is one-to-one, where }L=\mathcal{L}(P)\cup\mathcal{L}(Q))$$

      So, with the assumptions $(P'\parallel Q)[f]\sim_s P'[f]\parallel Q[f]$, $(P\parallel Q')[f]\sim_s P[f]\parallel Q'[f]$ and $(P'\parallel Q')[f]\sim_s P'[f]\parallel Q'[f]$, $(P\parallel Q)[f]\sim_s P[f]\parallel Q[f]$, if $f\upharpoonright(L\cup\overline{L})$ is one-to-one, where $L=\mathcal{L}(P)\cup\mathcal{L}(Q)$, as desired.
\end{enumerate}
\end{proof}

\begin{proposition}[Static laws for strong pomset bisimulation] \label{SLSPB}
The static laws for strong pomset bisimulation are as follows.
\begin{enumerate}
  \item $P\parallel Q\sim_p Q\parallel P$;
  \item $P\parallel(Q\parallel R)\sim_p (P\parallel Q)\parallel R$;
  \item $P\parallel \textbf{nil}\sim_p P$;
  \item $P\setminus L\sim_p P$, if $\mathcal{L}(P)\cap(L\cup\overline{L})=\emptyset$;
  \item $P\setminus K\setminus L\sim_p P\setminus(K\cup L)$;
  \item $P[f]\setminus L\sim_p P\setminus f^{-1}(L)[f]$;
  \item $(P\parallel Q)\setminus L\sim_p P\setminus L\parallel Q\setminus L$, if $\mathcal{L}(P)\cap\overline{\mathcal{L}(Q)}\cap(L\cup\overline{L})=\emptyset$;
  \item $P[Id]\sim_p P$;
  \item $P[f]\sim_p P[f']$, if $f\upharpoonright\mathcal{L}(P)=f'\upharpoonright\mathcal{L}(P)$;
  \item $P[f][f']\sim_p P[f'\circ f]$;
  \item $(P\parallel Q)[f]\sim_p P[f]\parallel Q[f]$, if $f\upharpoonright(L\cup\overline{L})$ is one-to-one, where $L=\mathcal{L}(P)\cup\mathcal{L}(Q)$.
\end{enumerate}
\end{proposition}

\begin{proof}
From the definition of strong pomset bisimulation (see Definition \ref{PSB}), we know that strong pomset bisimulation is defined by pomset transitions, which are labeled by pomsets. In a pomset transition, the events in the pomset are either within causality relations (defined by the prefix $.$) or in concurrency (implicitly defined by $.$ and $+$, and explicitly defined by $\parallel$), of course, they are pairwise consistent (without conflicts). In Proposition \ref{SLSSB}, we have already proven the case that all events are pairwise concurrent, so, we only need to prove the case of events in causality. Without loss of generality, we take a pomset of $p=\{\alpha,\beta:\alpha.\beta\}$. Then the pomset transition labeled by the above $p$ is just composed of one single event transition labeled by $\alpha$ succeeded by another single event transition labeled by $\beta$, that is, $\xrightarrow{p}=\xrightarrow{\alpha}\xrightarrow{\beta}$.

Similarly to the proof of static laws for strong step bisimulation (see Proposition \ref{SLSSB}), we can prove that the static laws hold for strong pomset bisimulation, we omit them.
\end{proof}

\begin{proposition}[Static laws for strong hp-bisimulation] \label{SLSHPB}
The static laws for strong hp-bisimulation are as follows.
\begin{enumerate}
  \item $P\parallel Q\sim_{hp} Q\parallel P$;
  \item $P\parallel(Q\parallel R)\sim_{hp} (P\parallel Q)\parallel R$;
  \item $P\parallel \textbf{nil}\sim_{hp} P$;
  \item $P\setminus L\sim_{hp} P$, if $\mathcal{L}(P)\cap(L\cup\overline{L})=\emptyset$;
  \item $P\setminus K\setminus L\sim_{hp} P\setminus(K\cup L)$;
  \item $P[f]\setminus L\sim_{hp} P\setminus f^{-1}(L)[f]$;
  \item $(P\parallel Q)\setminus L\sim_{hp} P\setminus L\parallel Q\setminus L$, if $\mathcal{L}(P)\cap\overline{\mathcal{L}(Q)}\cap(L\cup\overline{L})=\emptyset$;
  \item $P[Id]\sim_{hp} P$;
  \item $P[f]\sim_{hp} P[f']$, if $f\upharpoonright\mathcal{L}(P)=f'\upharpoonright\mathcal{L}(P)$;
  \item $P[f][f']\sim_{hp} P[f'\circ f]$;
  \item $(P\parallel Q)[f]\sim_{hp} P[f]\parallel Q[f]$, if $f\upharpoonright(L\cup\overline{L})$ is one-to-one, where $L=\mathcal{L}(P)\cup\mathcal{L}(Q)$.
\end{enumerate}
\end{proposition}

\begin{proof}
From the definition of strong hp-bisimulation (see Definition \ref{HHPB}), we know that strong hp-bisimulation is defined on the posetal product $(C_1,f,C_2),f:C_1\rightarrow C_2\textrm{ isomorphism}$. Two processes $P$ related to $C_1$ and $Q$ related to $C_2$, and $f:C_1\rightarrow C_2\textrm{ isomorphism}$. Initially, $(C_1,f,C_2)=(\emptyset,\emptyset,\emptyset)$, and $(\emptyset,\emptyset,\emptyset)\in\sim_{hp}$. When $P\xrightarrow{\alpha}P'$ ($C_1\xrightarrow{\alpha}C_1'$), there will be $Q\xrightarrow{\alpha}Q'$ ($C_2\xrightarrow{\alpha}C_2'$), and we define $f'=f[\alpha\mapsto \alpha]$. Then, if $(C_1,f,C_2)\in\sim_{hp}$, then $(C_1',f',C_2')\in\sim_{hp}$.

Similarly to the proof of static laws for strong pomset bisimulation (see Proposition \ref{SLSPB}), we can prove that static laws hold for strong hp-bisimulation, we just need additionally to check the above conditions on hp-bisimulation, we omit them.
\end{proof}

\begin{proposition}[Static laws for strongly hhp-bisimulation] \label{SLSHHPB}
The static laws for strongly hhp-bisimulation are as follows.
\begin{enumerate}
  \item $P\parallel Q\sim_{hhp} Q\parallel P$;
  \item $P\parallel(Q\parallel R)\sim_{hhp} (P\parallel Q)\parallel R$;
  \item $P\parallel \textbf{nil}\sim_{hhp} P$;
  \item $P\setminus L\sim_{hhp} P$, if $\mathcal{L}(P)\cap(L\cup\overline{L})=\emptyset$;
  \item $P\setminus K\setminus L\sim_{hhp} P\setminus(K\cup L)$;
  \item $P[f]\setminus L\sim_{hhp} P\setminus f^{-1}(L)[f]$;
  \item $(P\parallel Q)\setminus L\sim_{hhp} P\setminus L\parallel Q\setminus L$, if $\mathcal{L}(P)\cap\overline{\mathcal{L}(Q)}\cap(L\cup\overline{L})=\emptyset$;
  \item $P[Id]\sim_{hhp} P$;
  \item $P[f]\sim_{hhp} P[f']$, if $f\upharpoonright\mathcal{L}(P)=f'\upharpoonright\mathcal{L}(P)$;
  \item $P[f][f']\sim_{hhp} P[f'\circ f]$;
  \item $(P\parallel Q)[f]\sim_{hhp} P[f]\parallel Q[f]$, if $f\upharpoonright(L\cup\overline{L})$ is one-to-one, where $L=\mathcal{L}(P)\cup\mathcal{L}(Q)$.
\end{enumerate}
\end{proposition}

\begin{proof}
From the definition of strongly hhp-bisimulation (see Definition \ref{HHPB}), we know that strongly hhp-bisimulation is downward closed for strong hp-bisimulation.

Similarly to the proof of static laws for strong hp-bisimulation (see Proposition \ref{SLSHPB}), we can prove that static laws hold for strongly hhp-bisimulation, that is, they are downward closed for strong hp-bisimulation, we omit them.
\end{proof}

\begin{proposition}[Milner's expansion law for strongly truly concurrent bisimulations]
Milner's expansion law does not hold any more for any strongly truly concurrent bisimulation, that is,

\begin{enumerate}
  \item $\alpha\parallel\beta\nsim_p \alpha.\beta+\beta.\alpha$;
  \item $\alpha\parallel\beta\nsim_s \alpha.\beta+\beta.\alpha$;
  \item $\alpha\parallel\beta\nsim_{hp} \alpha.\beta+\beta.\alpha$;
  \item $\alpha\parallel\beta\nsim_{hhp} \alpha.\beta+\beta.\alpha$.
\end{enumerate}
\end{proposition}

\begin{proof}
In nature, it is caused by $\alpha\parallel \beta$ and $\alpha.\beta + \beta.\alpha$ having different causality structure. By the transition rules for $\textbf{Com}_{1,2,3,4}$, $\textbf{Sum}_{1,2}$ and $\textbf{Act}_{1,2}$, we have

$$\alpha\parallel \beta\xrightarrow{\{\alpha,\beta\}}\textbf{nil}$$

while

$$\alpha.\beta+ \beta.\alpha\nrightarrow^{\{\alpha,\beta\}}.$$
\end{proof}

\begin{proposition}[New expansion law for strong step bisimulation]\label{NELSSB}
Let $P\equiv (P_1[f_1]\parallel\cdots\parallel P_n[f_n])\setminus L$, with $n\geq 1$. Then

\begin{eqnarray}
P\sim_s \{(f_1(\alpha_1)\parallel\cdots\parallel f_n(\alpha_n)).(P_1'[f_1]\parallel\cdots\parallel P_n'[f_n])\setminus L: \nonumber\\
P_i\xrightarrow{\alpha_i}P_i',i\in\{1,\cdots,n\},f_i(\alpha_i)\notin L\cup\overline{L}\} \nonumber\\
+\sum\{\tau.(P_1[f_1]\parallel\cdots\parallel P_i'[f_i]\parallel\cdots\parallel P_j'[f_j]\parallel\cdots\parallel P_n[f_n])\setminus L: \nonumber\\
P_i\xrightarrow{l_1}P_i',P_j\xrightarrow{l_2}P_j',f_i(l_1)=\overline{f_j(l_2)},i<j\} \nonumber
\end{eqnarray}
\end{proposition}

\begin{proof}
Though transition rules in Table \ref{TRForCTC} are defined in the flavor of single event, they can be modified into a step (a set of events within which each event is pairwise concurrent), we omit them. If we treat a single event as a step containing just one event, the proof of the new expansion law has not any problem, so we use this way and still use the transition rules in Table \ref{TRForCTC}.

Firstly, we consider the case without Restriction and Relabeling. That is, we suffice to prove the following case by induction on the size $n$.

For $P\equiv P_1\parallel\cdots\parallel P_n$, with $n\geq 1$, we need to prove

\begin{eqnarray}
P\sim_s \{(\alpha_1\parallel\cdots\parallel \alpha_n).(P_1'\parallel\cdots\parallel P_n'): P_i\xrightarrow{\alpha_i}P_i',i\in\{1,\cdots,n\}\nonumber\\
+\sum\{\tau.(P_1\parallel\cdots\parallel P_i'\parallel\cdots\parallel P_j'\parallel\cdots\parallel P_n): P_i\xrightarrow{l}P_i',P_j\xrightarrow{\overline{l}}P_j',i<j\} \nonumber
\end{eqnarray}

For $n=1$, $P_1\sim_s \alpha_1.P_1':P_1\xrightarrow{\alpha_1}P_1'$ is obvious. Then with a hypothesis $n$, we consider $R\equiv P\parallel P_{n+1}$. By the transition rules $\textbf{Com}_{1,2,3,4}$, we can get

\begin{eqnarray}
R\sim_s \{(p\parallel \alpha_{n+1}).(P'\parallel P_{n+1}'): P\xrightarrow{p}P',P_{n+1}\xrightarrow{\alpha_{n+1}}P_{n+1}',p\subseteq P\}\nonumber\\
+\sum\{\tau.(P'\parallel P_{n+1}'): P\xrightarrow{l}P',P_{n+1}\xrightarrow{\overline{l}}P_{n+1}'\} \nonumber
\end{eqnarray}

Now with the induction assumption $P\equiv P_1\parallel\cdots\parallel P_n$, the right-hand side can be reformulated as follows.

\begin{eqnarray}
\{(\alpha_1\parallel\cdots\parallel \alpha_n\parallel \alpha_{n+1}).(P_1'\parallel\cdots\parallel P_n'\parallel P_{n+1}'): \nonumber\\
P_i\xrightarrow{\alpha_i}P_i',i\in\{1,\cdots,n+1\}\nonumber\\
+\sum\{\tau.(P_1\parallel\cdots\parallel P_i'\parallel\cdots\parallel P_j'\parallel\cdots\parallel P_n\parallel P_{n+1}): \nonumber\\
P_i\xrightarrow{l}P_i',P_j\xrightarrow{\overline{l}}P_j',i<j\} \nonumber\\
+\sum\{\tau.(P_1\parallel\cdots\parallel P_i'\parallel\cdots\parallel P_j\parallel\cdots\parallel P_n\parallel P_{n+1}'): \nonumber\\
P_i\xrightarrow{l}P_i',P_{n+1}\xrightarrow{\overline{l}}P_{n+1}',i\in\{1,\cdots, n\}\} \nonumber
\end{eqnarray}

So,

\begin{eqnarray}
R\sim_s \{(\alpha_1\parallel\cdots\parallel \alpha_n\parallel \alpha_{n+1}).(P_1'\parallel\cdots\parallel P_n'\parallel P_{n+1}'): \nonumber\\
P_i\xrightarrow{\alpha_i}P_i',i\in\{1,\cdots,n+1\}\nonumber\\
+\sum\{\tau.(P_1\parallel\cdots\parallel P_i'\parallel\cdots\parallel P_j'\parallel\cdots\parallel P_n): \nonumber\\
P_i\xrightarrow{l}P_i',P_j\xrightarrow{\overline{l}}P_j',1 \leq i<j\geq n+1\} \nonumber
\end{eqnarray}

Then, we can easily add the full conditions with Restriction and Relabeling.

\end{proof}

\begin{proposition}[New expansion law for strong pomset bisimulation]\label{NELSPB}
Let $P\equiv (P_1[f_1]\parallel\cdots\parallel P_n[f_n])\setminus L$, with $n\geq 1$. Then

\begin{eqnarray}
P\sim_p \{(f_1(\alpha_1)\parallel\cdots\parallel f_n(\alpha_n)).(P_1'[f_1]\parallel\cdots\parallel P_n'[f_n])\setminus L: \nonumber\\
P_i\xrightarrow{\alpha_i}P_i',i\in\{1,\cdots,n\},f_i(\alpha_i)\notin L\cup\overline{L}\} \nonumber\\
+\sum\{\tau.(P_1[f_1]\parallel\cdots\parallel P_i'[f_i]\parallel\cdots\parallel P_j'[f_j]\parallel\cdots\parallel P_n[f_n])\setminus L: \nonumber\\
P_i\xrightarrow{l_1}P_i',P_j\xrightarrow{l_2}P_j',f_i(l_1)=\overline{f_j(l_2)},i<j\} \nonumber
\end{eqnarray}
\end{proposition}

\begin{proof}
From the definition of strong pomset bisimulation (see Definition \ref{PSB}), we know that strong pomset bisimulation is defined by pomset transitions, which are labeled by pomsets. In a pomset transition, the events in the pomset are either within causality relations (defined by the prefix $.$) or in concurrency (implicitly defined by $.$ and $+$, and explicitly defined by $\parallel$), of course, they are pairwise consistent (without conflicts). In Proposition \ref{NELSSB}, we have already proven the case that all events are pairwise concurrent, so, we only need to prove the case of events in causality. Without loss of generality, we take a pomset of $p=\{\alpha,\beta:\alpha.\beta\}$. Then the pomset transition labeled by the above $p$ is just composed of one single event transition labeled by $\alpha$ succeeded by another single event transition labeled by $\beta$, that is, $\xrightarrow{p}=\xrightarrow{\alpha}\xrightarrow{\beta}$.

Similarly to the proof of new expansion law for strong step bisimulation (see Proposition \ref{NELSSB}), we can prove that the new expansion law holds for strong pomset bisimulation, we omit them.
\end{proof}

\begin{proposition}[New expansion law for strong hp-bisimulation]\label{NELSHPB}
Let $P\equiv (P_1[f_1]\parallel\cdots\parallel P_n[f_n])\setminus L$, with $n\geq 1$. Then

\begin{eqnarray}
P\sim_{hp} \{(f_1(\alpha_1)\parallel\cdots\parallel f_n(\alpha_n)).(P_1'[f_1]\parallel\cdots\parallel P_n'[f_n])\setminus L: \nonumber\\
P_i\xrightarrow{\alpha_i}P_i',i\in\{1,\cdots,n\},f_i(\alpha_i)\notin L\cup\overline{L}\} \nonumber\\
+\sum\{\tau.(P_1[f_1]\parallel\cdots\parallel P_i'[f_i]\parallel\cdots\parallel P_j'[f_j]\parallel\cdots\parallel P_n[f_n])\setminus L: \nonumber\\
P_i\xrightarrow{l_1}P_i',P_j\xrightarrow{l_2}P_j',f_i(l_1)=\overline{f_j(l_2)},i<j\} \nonumber
\end{eqnarray}
\end{proposition}

\begin{proof}
From the definition of strong hp-bisimulation (see Definition \ref{HHPB}), we know that strong hp-bisimulation is defined on the posetal product $(C_1,f,C_2),f:C_1\rightarrow C_2\textrm{ isomorphism}$. Two processes $P$ related to $C_1$ and $Q$ related to $C_2$, and $f:C_1\rightarrow C_2\textrm{ isomorphism}$. Initially, $(C_1,f,C_2)=(\emptyset,\emptyset,\emptyset)$, and $(\emptyset,\emptyset,\emptyset)\in\sim_{hp}$. When $P\xrightarrow{\alpha}P'$ ($C_1\xrightarrow{\alpha}C_1'$), there will be $Q\xrightarrow{\alpha}Q'$ ($C_2\xrightarrow{\alpha}C_2'$), and we define $f'=f[\alpha\mapsto \alpha]$. Then, if $(C_1,f,C_2)\in\sim_{hp}$, then $(C_1',f',C_2')\in\sim_{hp}$.

Similarly to the proof of new expansion law for strong pomset bisimulation (see Proposition \ref{NELSPB}), we can prove that the new expansion law holds for strong hp-bisimulation, we just need additionally to check the above conditions on hp-bisimulation, we omit them.
\end{proof}

\begin{proposition}[New expansion law for strongly hhp-bisimulation]\label{NELSHHPB}
Let $P\equiv (P_1[f_1]\parallel\cdots\parallel P_n[f_n])\setminus L$, with $n\geq 1$. Then

\begin{eqnarray}
P\sim_{hhp} \{(f_1(\alpha_1)\parallel\cdots\parallel f_n(\alpha_n)).(P_1'[f_1]\parallel\cdots\parallel P_n'[f_n])\setminus L: \nonumber\\
P_i\xrightarrow{\alpha_i}P_i',i\in\{1,\cdots,n\},f_i(\alpha_i)\notin L\cup\overline{L}\} \nonumber\\
+\sum\{\tau.(P_1[f_1]\parallel\cdots\parallel P_i'[f_i]\parallel\cdots\parallel P_j'[f_j]\parallel\cdots\parallel P_n[f_n])\setminus L: \nonumber\\
P_i\xrightarrow{l_1}P_i',P_j\xrightarrow{l_2}P_j',f_i(l_1)=\overline{f_j(l_2)},i<j\} \nonumber
\end{eqnarray}
\end{proposition}

\begin{proof}
From the definition of strongly hhp-bisimulation (see Definition \ref{HHPB}), we know that strongly hhp-bisimulation is downward closed for strong hp-bisimulation.

Similarly to the proof of the new expansion law for strong hp-bisimulation (see Proposition \ref{NELSHPB}), we can prove that the new expansion law holds for strongly hhp-bisimulation, that is, they are downward closed for strong hp-bisimulation, we omit them.
\end{proof}

\begin{theorem}[Congruence for strong step bisimulation] \label{CSSB}
We can enjoy the full congruence for strong step bisimulation as follows.
\begin{enumerate}
  \item If $A\overset{\text{def}}{=}P$, then $A\sim_s P$;
  \item Let $P_1\sim_s P_2$. Then
        \begin{enumerate}
           \item $\alpha.P_1\sim_s \alpha.P_2$;
           \item $(\alpha_1\parallel\cdots\parallel\alpha_n).P_1\sim_s (\alpha_1\parallel\cdots\parallel\alpha_n).P_2$;
           \item $P_1+Q\sim_s P_2 +Q$;
           \item $P_1\parallel Q\sim_s P_2\parallel Q$;
           \item $P_1\setminus L\sim_s P_2\setminus L$;
           \item $P_1[f]\sim_s P_2[f]$.
         \end{enumerate}
\end{enumerate}
\end{theorem}

\begin{proof}
Though transition rules in Table \ref{TRForCTC} are defined in the flavor of single event, they can be modified into a step (a set of events within which each event is pairwise concurrent), we omit them. If we treat a single event as a step containing just one event, the proof of the congruence does not exist any problem, so we use this way and still use the transition rules in Table \ref{TRForCTC}.

\begin{enumerate}
  \item If $A\overset{\text{def}}{=}P$, then $A\sim_s P$. It is obvious.
  \item Let $P_1\sim_s P_2$. Then
        \begin{enumerate}
           \item $\alpha.P_1\sim_s \alpha.P_2$. By the transition rules of $\textbf{Act}_{1,2}$ in Table \ref{TRForCTC}, we can get

           $$\alpha.P_1\xrightarrow{\alpha}P_1$$

           $$\alpha.P_2\xrightarrow{\alpha}P_2$$

           Since $P_1\sim_s P_2$, we get $\alpha.P_1\sim_s \alpha.P_2$, as desired.
           \item $(\alpha_1\parallel\cdots\parallel\alpha_n).P_1\sim_s (\alpha_1\parallel\cdots\parallel\alpha_n).P_2$. By the transition rules of $\textbf{Act}_{1,2}$ in Table \ref{TRForCTC}, we can get

           $$(\alpha_1\parallel\cdots\parallel\alpha_n).P_1\xrightarrow{\{\alpha_1,\cdots,\alpha_n\}}P_1$$

           $$(\alpha_1\parallel\cdots\parallel\alpha_n).P_2\xrightarrow{\{\alpha_1,\cdots,\alpha_n\}}P_2$$

           Since $P_1\sim_s P_2$, we get $(\alpha_1\parallel\cdots\parallel\alpha_n).P_1\sim_s (\alpha_1\parallel\cdots\parallel\alpha_n).P_2$, as desired.
           \item $P_1+Q\sim_s P_2 +Q$. By the transition rules of $\textbf{Sum}_{1,2}$ in Table \ref{TRForCTC}, we can get

           $$\frac{P_1\xrightarrow{\alpha}P_1'}{P_2\xrightarrow{\alpha}P_2'}(P_1'\sim_s P_2')$$

           $$\frac{P_1\xrightarrow{\alpha}P_1'}{P_1+Q\xrightarrow{\alpha}P_1'}
           \quad \frac{P_2\xrightarrow{\alpha}P_2'}{P_2+Q\xrightarrow{\alpha}P_2'}$$

           $$\frac{Q\xrightarrow{\beta}Q'}{P_1+Q\xrightarrow{\beta}Q'}
           \quad \frac{Q\xrightarrow{\beta}Q'}{P_2+Q\xrightarrow{\beta}Q'}$$

           Since $P_1'\sim_s P_2'$ and $Q'\sim_s Q'$, we get $P_1+Q\sim_s P_2+Q$, as desired.
           \item $P_1\parallel Q\sim_s P_2\parallel Q$. By the transition rules of $\textbf{Com}_{1,2,3,4}$ in Table \ref{TRForCTC}, we can get

           $$\frac{P_1\xrightarrow{\alpha}P_1'}{P_2\xrightarrow{\alpha}P_2'}(P_1'\sim_s P_2')$$

           $$\frac{P_1\xrightarrow{\alpha}P_1'\quad Q\nrightarrow}{P_1\parallel Q\xrightarrow{\alpha}P_1'\parallel Q}
           \quad \frac{P_2\xrightarrow{\alpha}P_2'\quad Q\nrightarrow}{P_2\parallel Q\xrightarrow{\alpha}P_2'\parallel Q}$$

           $$\frac{Q\xrightarrow{\beta}Q'\quad P_1\nrightarrow}{P_1\parallel Q\xrightarrow{\beta}P_1\parallel Q'}
           \quad \frac{Q\xrightarrow{\beta}P_2'\quad P_2\nrightarrow}{P_2\parallel Q\xrightarrow{\beta}P_2\parallel Q'}$$

           $$\frac{P_1\xrightarrow{\alpha}P_1'\quad Q\xrightarrow{\beta}Q'}{P_1\parallel Q\xrightarrow{\{\alpha,\beta\}}P_1'\parallel Q'}(\beta\neq\overline{\alpha})
           \quad \frac{P_2\xrightarrow{\alpha}P_2'\quad Q\xrightarrow{\beta}Q'}{P_2\parallel Q\xrightarrow{\{\alpha,\beta\}}P_2'\parallel Q'}(\beta\neq\overline{\alpha})$$

           $$\frac{P_1\xrightarrow{l}P_1'\quad Q\xrightarrow{\overline{l}}Q'}{P_1\parallel Q\xrightarrow{\tau}P_1'\parallel Q'}
           \quad \frac{P_2\xrightarrow{l}P_2'\quad Q\xrightarrow{\overline{l}}Q'}{P_2\parallel Q\xrightarrow{\tau}P_2'\parallel Q'}$$

           Since $P_1'\sim_s P_2'$ and $Q'\sim_s Q'$, and with the assumptions $P_1'\parallel Q\sim_s P_2'\parallel Q$, $P_1\parallel Q'\sim_s P_2\parallel Q'$ and $P_1'\parallel Q'\sim_s P_2'\parallel Q'$, we get $P_1\parallel Q\sim_s P_2\parallel Q$, as desired.
           \item $P_1\setminus L\sim_s P_2\setminus L$. By the transition rules of $\textbf{Res}_{1,2}$ in Table \ref{TRForCTC}, we get

           $$\frac{P_1\xrightarrow{\alpha}P_1'}{P_2\xrightarrow{\alpha}P_2'}(P_1'\sim_s P_2')$$

           $$\frac{P_1\xrightarrow{\alpha}P_1'}{P_1\setminus L\xrightarrow{\alpha}P_1'\setminus L}$$

           $$\frac{P_2\xrightarrow{\alpha}P_2'}{P_2\setminus L\xrightarrow{\alpha}P_2'\setminus L}$$

           Since $P_1'\sim_s P_2'$, and with the assumption $P_1'\setminus L\sim_s P_2'\setminus L$, we get $P_1\setminus L\sim_s P_2\setminus L$, as desired.
           \item $P_1[f]\sim_s P_2[f]$. By the transition rules of $\textbf{Rel}_{1,2}$ in Table \ref{TRForCTC}, we get

           $$\frac{P_1\xrightarrow{\alpha}P_1'}{P_2\xrightarrow{\alpha}P_2'}(P_1'\sim_s P_2')$$

           $$\frac{P_1\xrightarrow{\alpha}P_1'}{P_1[f]\xrightarrow{f(\alpha)}P_1'[f]}$$

           $$\frac{P_2\xrightarrow{\alpha}P_2'}{P_2[f]\xrightarrow{f(\alpha)}P_2'[f]}$$

           Since $P_1'\sim_s P_2'$, and with the assumption $P_1'[f]\sim_s P_2'[f]$, we get $P_1[f]\sim_s P_2[f]$, as desired.
         \end{enumerate}
\end{enumerate}
\end{proof}

\begin{theorem}[Congruence for strong pomset bisimulation] \label{CSPB}
We can enjoy the full congruence for strong pomset bisimulation as follows.
\begin{enumerate}
  \item If $A\overset{\text{def}}{=}P$, then $A\sim_p P$;
  \item Let $P_1\sim_p P_2$. Then
        \begin{enumerate}
           \item $\alpha.P_1\sim_p \alpha.P_2$;
           \item $(\alpha_1\parallel\cdots\parallel\alpha_n).P_1\sim_p (\alpha_1\parallel\cdots\parallel\alpha_n).P_2$;
           \item $P_1+Q\sim_p P_2 +Q$;
           \item $P_1\parallel Q\sim_p P_2\parallel Q$;
           \item $P_1\setminus L\sim_p P_2\setminus L$;
           \item $P_1[f]\sim_p P_2[f]$.
         \end{enumerate}
\end{enumerate}
\end{theorem}

\begin{proof}
From the definition of strong pomset bisimulation (see Definition \ref{PSB}), we know that strong pomset bisimulation is defined by pomset transitions, which are labeled by pomsets. In a pomset transition, the events in the pomset are either within causality relations (defined by the prefix $.$) or in concurrency (implicitly defined by $.$ and $+$, and explicitly defined by $\parallel$), of course, they are pairwise consistent (without conflicts). In Theorem \ref{CSSB}, we have already proven the case that all events are pairwise concurrent, so, we only need to prove the case of events in causality. Without loss of generality, we take a pomset of $p=\{\alpha,\beta:\alpha.\beta\}$. Then the pomset transition labeled by the above $p$ is just composed of one single event transition labeled by $\alpha$ succeeded by another single event transition labeled by $\beta$, that is, $\xrightarrow{p}=\xrightarrow{\alpha}\xrightarrow{\beta}$.

Similarly to the proof of congruence for strong step bisimulation (see Theorem \ref{CSSB}), we can prove that the congruence holds for strong pomset bisimulation, we omit them.
\end{proof}

\begin{theorem}[Congruence for strong hp-bisimulation] \label{CSHPB}
We can enjoy the full congruence for strong hp-bisimulation as follows.
\begin{enumerate}
  \item If $A\overset{\text{def}}{=}P$, then $A\sim_{hp} P$;
  \item Let $P_1\sim_{hp} P_2$. Then
        \begin{enumerate}
           \item $\alpha.P_1\sim_{hp} \alpha.P_2$;
           \item $(\alpha_1\parallel\cdots\parallel\alpha_n).P_1\sim_{hp} (\alpha_1\parallel\cdots\parallel\alpha_n).P_2$;
           \item $P_1+Q\sim_{hp} P_2 +Q$;
           \item $P_1\parallel Q\sim_{hp} P_2\parallel Q$;
           \item $P_1\setminus L\sim_{hp} P_2\setminus L$;
           \item $P_1[f]\sim_{hp} P_2[f]$.
         \end{enumerate}
\end{enumerate}
\end{theorem}

\begin{proof}
From the definition of strong hp-bisimulation (see Definition \ref{HHPB}), we know that strong hp-bisimulation is defined on the posetal product $(C_1,f,C_2),f:C_1\rightarrow C_2\textrm{ isomorphism}$. Two processes $P$ related to $C_1$ and $Q$ related to $C_2$, and $f:C_1\rightarrow C_2\textrm{ isomorphism}$. Initially, $(C_1,f,C_2)=(\emptyset,\emptyset,\emptyset)$, and $(\emptyset,\emptyset,\emptyset)\in\sim_{hp}$. When $P\xrightarrow{\alpha}P'$ ($C_1\xrightarrow{\alpha}C_1'$), there will be $Q\xrightarrow{\alpha}Q'$ ($C_2\xrightarrow{\alpha}C_2'$), and we define $f'=f[\alpha\mapsto \alpha]$. Then, if $(C_1,f,C_2)\in\sim_{hp}$, then $(C_1',f',C_2')\in\sim_{hp}$.

Similarly to the proof of congruence for strong pomset bisimulation (see Theorem \ref{CSPB}), we can prove that the congruence holds for strong hp-bisimulation, we just need additionally to check the above conditions on hp-bisimulation, we omit them.
\end{proof}

\begin{theorem}[Congruence for strongly hhp-bisimulation] \label{CSHHPB}
We can enjoy the full congruence for strongly hhp-bisimulation as follows.
\begin{enumerate}
  \item If $A\overset{\text{def}}{=}P$, then $A\sim_{hhp} P$;
  \item Let $P_1\sim_{hhp} P_2$. Then
        \begin{enumerate}
           \item $\alpha.P_1\sim_{hhp} \alpha.P_2$;
           \item $(\alpha_1\parallel\cdots\parallel\alpha_n).P_1\sim_{hhp} (\alpha_1\parallel\cdots\parallel\alpha_n).P_2$;
           \item $P_1+Q\sim_{hhp} P_2 +Q$;
           \item $P_1\parallel Q\sim_{hhp} P_2\parallel Q$;
           \item $P_1\setminus L\sim_{hhp} P_2\setminus L$;
           \item $P_1[f]\sim_{hhp} P_2[f]$.
         \end{enumerate}
\end{enumerate}
\end{theorem}

\begin{proof}
From the definition of strongly hhp-bisimulation (see Definition \ref{HHPB}), we know that strongly hhp-bisimulation is downward closed for strong hp-bisimulation.

Similarly to the proof of congruence for strong hp-bisimulation (see Theorem \ref{CSHPB}), we can prove that the congruence holds for strongly hhp-bisimulation, we omit them.
\end{proof}

\subsection{Recursion}

\begin{definition}[Weakly guarded recursive expression]
$X$ is weakly guarded in $E$ if each occurrence of $X$ is with some subexpression $\alpha.F$ or $(\alpha_1\parallel\cdots\parallel\alpha_n).F$ of $E$.
\end{definition}

\begin{lemma}\label{LUS}
If the variables $\widetilde{X}$ are weakly guarded in $E$, and $E\{\widetilde{P}/\widetilde{X}\}\xrightarrow{\{\alpha_1,\cdots,\alpha_n\}}P'$, then $P'$ takes the form $E'\{\widetilde{P}/\widetilde{X}\}$ for some expression $E'$, and moreover, for any $\widetilde{Q}$, $E\{\widetilde{Q}/\widetilde{X}\}\xrightarrow{\{\alpha_1,\cdots,\alpha_n\}}E'\{\widetilde{Q}/\widetilde{X}\}$.
\end{lemma}

\begin{proof}
It needs to induct on the depth of the inference of $E\{\widetilde{P}/\widetilde{X}\}\xrightarrow{\{\alpha_1,\cdots,\alpha_n\}}P'$.

\begin{enumerate}
  \item Case $E\equiv Y$, a variable. Then $Y\notin \widetilde{X}$. Since $\widetilde{X}$ are weakly guarded, $Y\{\widetilde{P}/\widetilde{X}\equiv Y\}\nrightarrow$, this case is impossible.
  \item Case $E\equiv\beta.F$. Then we must have $\alpha=\beta$, and $P'\equiv F\{\widetilde{P}/\widetilde{X}\}$, and $E\{\widetilde{Q}/\widetilde{X}\}\equiv \beta.F\{\widetilde{Q}/\widetilde{X}\} \xrightarrow{\beta}F\{\widetilde{Q}/\widetilde{X}\}$, then, let $E'$ be $F$, as desired.
  \item Case $E\equiv(\beta_1\parallel\cdots\parallel\beta_n).F$. Then we must have $\alpha_i=\beta_i$ for $1\leq i\leq n$, and $P'\equiv F\{\widetilde{P}/\widetilde{X}\}$, and $E\{\widetilde{Q}/\widetilde{X}\}\equiv (\beta_1\parallel\cdots\parallel\beta_n).F\{\widetilde{Q}/\widetilde{X}\} \xrightarrow{\{\beta_1,\cdots,\beta_n\}}F\{\widetilde{Q}/\widetilde{X}\}$, then, let $E'$ be $F$, as desired.
  \item Case $E\equiv E_1+E_2$. Then either $E_1\{\widetilde{P}/\widetilde{X}\} \xrightarrow{\{\alpha_1,\cdots,\alpha_n\}}P'$ or $E_2\{\widetilde{P}/\widetilde{X}\} \xrightarrow{\{\alpha_1,\cdots,\alpha_n\}}P'$, then, we can apply this lemma in either case, as desired.
  \item Case $E\equiv E_1\parallel E_2$. There are four possibilities.
  \begin{enumerate}
    \item We may have $E_1\{\widetilde{P}/\widetilde{X}\} \xrightarrow{\alpha}P_1'$ and $E_2\{\widetilde{P}/\widetilde{X}\}\nrightarrow$ with $P'\equiv P_1'\parallel (E_2\{\widetilde{P}/\widetilde{X}\})$, then by applying this lemma, $P_1'$ is of the form $E_1'\{\widetilde{P}/\widetilde{X}\}$, and for any $Q$, $E_1\{\widetilde{Q}/\widetilde{X}\}\xrightarrow{\alpha} E_1'\{\widetilde{Q}/\widetilde{X}\}$. So, $P'$ is of the form $E_1'\parallel E_2\{\widetilde{P}/\widetilde{X}\}$, and for any $Q$, $E\{\widetilde{Q}/\widetilde{X}\}\equiv E_1\{\widetilde{Q}/\widetilde{X}\}\parallel E_2\{\widetilde{Q}/\widetilde{X}\}\xrightarrow{\alpha} (E_1'\parallel E_2)\{\widetilde{Q}/\widetilde{X}\}$, then, let $E'$ be $E_1'\parallel E_2$, as desired.
    \item We may have $E_2\{\widetilde{P}/\widetilde{X}\} \xrightarrow{\alpha}P_2'$ and $E_1\{\widetilde{P}/\widetilde{X}\}\nrightarrow$ with $P'\equiv P_2'\parallel (E_1\{\widetilde{P}/\widetilde{X}\})$, this case can be prove similarly to the above subcase, as desired.
    \item We may have $E_1\{\widetilde{P}/\widetilde{X}\} \xrightarrow{\alpha}P_1'$ and $E_2\{\widetilde{P}/\widetilde{X}\}\xrightarrow{\beta}P_2'$ with $\alpha\neq\overline{\beta}$ and $P'\equiv P_1'\parallel P_2'$, then by applying this lemma, $P_1'$ is of the form $E_1'\{\widetilde{P}/\widetilde{X}\}$, and for any $Q$, $E_1\{\widetilde{Q}/\widetilde{X}\}\xrightarrow{\alpha} E_1'\{\widetilde{Q}/\widetilde{X}\}$; $P_2'$ is of the form $E_2'\{\widetilde{P}/\widetilde{X}\}$, and for any $Q$, $E_2\{\widetilde{Q}/\widetilde{X}\}\xrightarrow{\alpha} E_2'\{\widetilde{Q}/\widetilde{X}\}$. So, $P'$ is of the form $E_1'\parallel E_2'\{\widetilde{P}/\widetilde{X}\}$, and for any $Q$, $E\{\widetilde{Q}/\widetilde{X}\}\equiv E_1\{\widetilde{Q}/\widetilde{X}\}\parallel E_2\{\widetilde{Q}/\widetilde{X}\}\xrightarrow{\{\alpha,\beta\}} (E_1'\parallel E_2')\{\widetilde{Q}/\widetilde{X}\}$, then, let $E'$ be $E_1'\parallel E_2'$, as desired.
    \item We may have $E_1\{\widetilde{P}/\widetilde{X}\} \xrightarrow{l}P_1'$ and $E_2\{\widetilde{P}/\widetilde{X}\}\xrightarrow{\overline{l}}P_2'$ with $P'\equiv P_1'\parallel P_2'$, then by applying this lemma, $P_1'$ is of the form $E_1'\{\widetilde{P}/\widetilde{X}\}$, and for any $Q$, $E_1\{\widetilde{Q}/\widetilde{X}\}\xrightarrow{l} E_1'\{\widetilde{Q}/\widetilde{X}\}$; $P_2'$ is of the form $E_2'\{\widetilde{P}/\widetilde{X}\}$, and for any $Q$, $E_2\{\widetilde{Q}/\widetilde{X}\}\xrightarrow{\overline{l}} E_2'\{\widetilde{Q}/\widetilde{X}\}$. So, $P'$ is of the form $E_1'\parallel E_2'\{\widetilde{P}/\widetilde{X}\}$, and for any $Q$, $E\{\widetilde{Q}/\widetilde{X}\}\equiv E_1\{\widetilde{Q}/\widetilde{X}\}\parallel E_2\{\widetilde{Q}/\widetilde{X}\}\xrightarrow{\tau} (E_1'\parallel E_2')\{\widetilde{Q}/\widetilde{X}\}$, then, let $E'$ be $E_1'\parallel E_2'$, as desired.
  \end{enumerate}
  \item Case $E\equiv F[R]$ and $E\equiv F\setminus L$. These cases can be prove similarly to the above case.
  \item Case $E\equiv C$, an agent constant defined by $C\overset{\text{def}}{=}R$. Then there is no $X\in\widetilde{X}$ occurring in $E$, so $C\xrightarrow{\{\alpha_1,\cdots,\alpha_n\}}P'$, let $E'$ be $P'$, as desired.
\end{enumerate}
\end{proof}

\begin{theorem}[Unique solution of equations for strong step bisimulation]\label{USSSB}
Let the recursive expressions $E_i(i\in I)$ contain at most the variables $X_i(i\in I)$, and let each $X_j(j\in I)$ be weakly guarded in each $E_i$. Then,

If $\widetilde{P}\sim_s \widetilde{E}\{\widetilde{P}/\widetilde{X}\}$ and $\widetilde{Q}\sim_s \widetilde{E}\{\widetilde{Q}/\widetilde{X}\}$, then $\widetilde{P}\sim_s \widetilde{Q}$.
\end{theorem}

\begin{proof}
It is sufficient to induct on the depth of the inference of $E\{\widetilde{P}/\widetilde{X}\}\xrightarrow{\{\alpha_1,\cdots,\alpha_n\}}P'$.

\begin{enumerate}
  \item Case $E\equiv X_i$. Then we have $E\{\widetilde{P}/\widetilde{X}\}\equiv P_i\xrightarrow{\{\alpha_1,\cdots,\alpha_n\}}P'$, since $P_i\sim_s E_i\{\widetilde{P}/\widetilde{X}\}$, we have $E_i\{\widetilde{P}/\widetilde{X}\}\xrightarrow{\{\alpha_1,\cdots,\alpha_n\}}P''\sim_s P'$. Since $\widetilde{X}$ are weakly guarded in $E_i$, by Lemma \ref{LUS}, $P''\equiv E'\{\widetilde{P}/\widetilde{X}\}$ and $E_i\{\widetilde{P}/\widetilde{X}\}\xrightarrow{\{\alpha_1,\cdots,\alpha_n\}} E'\{\widetilde{P}/\widetilde{X}\}$. Since $E\{\widetilde{Q}/\widetilde{X}\}\equiv X_i\{\widetilde{Q}/\widetilde{X}\} \equiv Q_i\sim_s E_i\{\widetilde{Q}/\widetilde{X}\}$, $E\{\widetilde{Q}/\widetilde{X}\}\xrightarrow{\{\alpha_1,\cdots,\alpha_n\}}Q'\sim_s E'\{\widetilde{Q}/\widetilde{X}\}$. So, $P'\sim_s Q'$, as desired.
  \item Case $E\equiv\alpha.F$. This case can be proven similarly.
  \item Case $E\equiv(\alpha_1\parallel\cdots\parallel\alpha_n).F$. This case can be proven similarly.
  \item Case $E\equiv E_1+E_2$. We have $E_i\{\widetilde{P}/\widetilde{X}\} \xrightarrow{\{\alpha_1,\cdots,\alpha_n\}}P'$, $E_i\{\widetilde{Q}/\widetilde{X}\} \xrightarrow{\{\alpha_1,\cdots,\alpha_n\}}Q'$, then, $P'\sim_s Q'$, as desired.
  \item Case $E\equiv E_1\parallel E_2$, $E\equiv F[R]$ and $E\equiv F\setminus L$, $E\equiv C$. These cases can be prove similarly to the above case.
\end{enumerate}
\end{proof}

\begin{theorem}[Unique solution of equations for strong pomset bisimulation]\label{USSPB}
Let the recursive expressions $E_i(i\in I)$ contain at most the variables $X_i(i\in I)$, and let each $X_j(j\in I)$ be weakly guarded in each $E_i$. Then,

If $\widetilde{P}\sim_p \widetilde{E}\{\widetilde{P}/\widetilde{X}\}$ and $\widetilde{Q}\sim_p \widetilde{E}\{\widetilde{Q}/\widetilde{X}\}$, then $\widetilde{P}\sim_p \widetilde{Q}$.
\end{theorem}

\begin{proof}
From the definition of strong pomset bisimulation (see Definition \ref{PSB}), we know that strong pomset bisimulation is defined by pomset transitions, which are labeled by pomsets. In a pomset transition, the events in the pomset are either within causality relations (defined by the prefix $.$) or in concurrency (implicitly defined by $.$ and $+$, and explicitly defined by $\parallel$), of course, they are pairwise consistent (without conflicts). In Theorem \ref{USSSB}, we have already proven the case that all events are pairwise concurrent, so, we only need to prove the case of events in causality. Without loss of generality, we take a pomset of $p=\{\alpha,\beta:\alpha.\beta\}$. Then the pomset transition labeled by the above $p$ is just composed of one single event transition labeled by $\alpha$ succeeded by another single event transition labeled by $\beta$, that is, $\xrightarrow{p}=\xrightarrow{\alpha}\xrightarrow{\beta}$.

Similarly to the proof of unique solution of equations for strong step bisimulation (see Theorem \ref{USSSB}), we can prove that the unique solution of equations holds for strong pomset bisimulation, we omit them.
\end{proof}

\begin{theorem}[Unique solution of equations for strong hp-bisimulation]\label{USSHPB}
Let the recursive expressions $E_i(i\in I)$ contain at most the variables $X_i(i\in I)$, and let each $X_j(j\in I)$ be weakly guarded in each $E_i$. Then,

If $\widetilde{P}\sim_{hp} \widetilde{E}\{\widetilde{P}/\widetilde{X}\}$ and $\widetilde{Q}\sim_{hp} \widetilde{E}\{\widetilde{Q}/\widetilde{X}\}$, then $\widetilde{P}\sim_{hp} \widetilde{Q}$.
\end{theorem}

\begin{proof}
From the definition of strong hp-bisimulation (see Definition \ref{HHPB}), we know that strong hp-bisimulation is defined on the posetal product $(C_1,f,C_2),f:C_1\rightarrow C_2\textrm{ isomorphism}$. Two processes $P$ related to $C_1$ and $Q$ related to $C_2$, and $f:C_1\rightarrow C_2\textrm{ isomorphism}$. Initially, $(C_1,f,C_2)=(\emptyset,\emptyset,\emptyset)$, and $(\emptyset,\emptyset,\emptyset)\in\sim_{hp}$. When $P\xrightarrow{\alpha}P'$ ($C_1\xrightarrow{\alpha}C_1'$), there will be $Q\xrightarrow{\alpha}Q'$ ($C_2\xrightarrow{\alpha}C_2'$), and we define $f'=f[\alpha\mapsto \alpha]$. Then, if $(C_1,f,C_2)\in\sim_{hp}$, then $(C_1',f',C_2')\in\sim_{hp}$.

Similarly to the proof of unique solution of equations for strong pomset bisimulation (see Theorem \ref{USSPB}), we can prove that the unique solution of equations holds for strong hp-bisimulation, we just need additionally to check the above conditions on hp-bisimulation, we omit them.
\end{proof}

\begin{theorem}[Unique solution of equations for strongly hhp-bisimulation]\label{USSHHPB}
Let the recursive expressions $E_i(i\in I)$ contain at most the variables $X_i(i\in I)$, and let each $X_j(j\in I)$ be weakly guarded in each $E_i$. Then,

If $\widetilde{P}\sim_{hhp} \widetilde{E}\{\widetilde{P}/\widetilde{X}\}$ and $\widetilde{Q}\sim_{hhp} \widetilde{E}\{\widetilde{Q}/\widetilde{X}\}$, then $\widetilde{P}\sim_{hhp} \widetilde{Q}$.
\end{theorem}

\begin{proof}
From the definition of strongly hhp-bisimulation (see Definition \ref{HHPB}), we know that strongly hhp-bisimulation is downward closed for strong hp-bisimulation.

Similarly to the proof of unique solution of equations for strong hp-bisimulation (see Theorem \ref{USSHPB}), we can prove that the unique solution of equations holds for strongly hhp-bisimulation, we omit them.
\end{proof}

\section{Weakly Truly Concurrent Bisimulations}\label{wtcb}

\subsection{Basic Definitions}\label{WTCC}

In this subsection, we introduce several weakly truly concurrent bisimulation equivalences, including weak pomset bisimulation, weak step bisimulation, weak history-preserving (hp-)bisimulation and weakly hereditary history-preserving (hhp-)bisimulation.

\begin{definition}[Weak pomset transitions and weak step]\label{WPT}
Let $\mathcal{E}$ be a PES and let $C\in\mathcal{C}(\mathcal{E})$, and $\emptyset\neq X\subseteq \hat{\mathbb{E}}$, if $C\cap X=\emptyset$ and $\hat{C'}=\hat{C}\cup X\in\mathcal{C}(\mathcal{E})$, then $C\xRightarrow{X} C'$ is called a weak pomset transition from $C$ to $C'$, where we define $\xRightarrow{e}\triangleq\xrightarrow{\tau^*}\xrightarrow{e}\xrightarrow{\tau^*}$. And $\xRightarrow{X}\triangleq\xrightarrow{\tau^*}\xrightarrow{e}\xrightarrow{\tau^*}$, for every $e\in X$. When the events in $X$ are pairwise concurrent, we say that $C\xRightarrow{X}C'$ is a weak step.
\end{definition}

\begin{definition}[Weak pomset, step bisimulation]\label{WPSB}
Let $\mathcal{E}_1$, $\mathcal{E}_2$ be PESs. A weak pomset bisimulation is a relation $R\subseteq\mathcal{C}(\mathcal{E}_1)\times\mathcal{C}(\mathcal{E}_2)$, such that if $(C_1,C_2)\in R$, and $C_1\xRightarrow{X_1}C_1'$ then $C_2\xRightarrow{X_2}C_2'$, with $X_1\subseteq \hat{\mathbb{E}_1}$, $X_2\subseteq \hat{\mathbb{E}_2}$, $X_1\sim X_2$ and $(C_1',C_2')\in R$, and vice-versa. We say that $\mathcal{E}_1$, $\mathcal{E}_2$ are weak pomset bisimilar, written $\mathcal{E}_1\approx_p\mathcal{E}_2$, if there exists a weak pomset bisimulation $R$, such that $(\emptyset,\emptyset)\in R$. By replacing weak pomset transitions with weak steps, we can get the definition of weak step bisimulation. When PESs $\mathcal{E}_1$ and $\mathcal{E}_2$ are weak step bisimilar, we write $\mathcal{E}_1\approx_s\mathcal{E}_2$.
\end{definition}

\begin{definition}[Weakly posetal product]
Given two PESs $\mathcal{E}_1$, $\mathcal{E}_2$, the weakly posetal product of their configurations, denoted $\mathcal{C}(\mathcal{E}_1)\overline{\times}\mathcal{C}(\mathcal{E}_2)$, is defined as

$$\{(C_1,f,C_2)|C_1\in\mathcal{C}(\mathcal{E}_1),C_2\in\mathcal{C}(\mathcal{E}_2),f:\hat{C_1}\rightarrow \hat{C_2} \textrm{ isomorphism}\}.$$

A subset $R\subseteq\mathcal{C}(\mathcal{E}_1)\overline{\times}\mathcal{C}(\mathcal{E}_2)$ is called a weakly posetal relation. We say that $R$ is downward closed when for any $(C_1,f,C_2),(C_1',f,C_2')\in \mathcal{C}(\mathcal{E}_1)\overline{\times}\mathcal{C}(\mathcal{E}_2)$, if $(C_1,f,C_2)\subseteq (C_1',f',C_2')$ pointwise and $(C_1',f',C_2')\in R$, then $(C_1,f,C_2)\in R$.

For $f:X_1\rightarrow X_2$, we define $f[x_1\mapsto x_2]:X_1\cup\{x_1\}\rightarrow X_2\cup\{x_2\}$, $z\in X_1\cup\{x_1\}$,(1)$f[x_1\mapsto x_2](z)=
x_2$,if $z=x_1$;(2)$f[x_1\mapsto x_2](z)=f(z)$, otherwise. Where $X_1\subseteq \hat{\mathbb{E}_1}$, $X_2\subseteq \hat{\mathbb{E}_2}$, $x_1\in \hat{\mathbb{E}}_1$, $x_2\in \hat{\mathbb{E}}_2$. Also, we define $f(\tau^*)=f(\tau^*)$.
\end{definition}

\begin{definition}[Weak (hereditary) history-preserving bisimulation]\label{WHHPB}
A weak history-preserving (hp-) bisimulation is a weakly posetal relation $R\subseteq\mathcal{C}(\mathcal{E}_1)\overline{\times}\mathcal{C}(\mathcal{E}_2)$ such that if $(C_1,f,C_2)\in R$, and $C_1\xRightarrow{e_1} C_1'$, then $C_2\xRightarrow{e_2} C_2'$, with $(C_1',f[e_1\mapsto e_2],C_2')\in R$, and vice-versa. $\mathcal{E}_1,\mathcal{E}_2$ are weak history-preserving (hp-)bisimilar and are written $\mathcal{E}_1\approx_{hp}\mathcal{E}_2$ if there exists a hp-bisimulation $R$ such that $(\emptyset,\emptyset,\emptyset)\in R$.

A weakly hereditary history-preserving (hhp-)bisimulation is a downward closed weak hp-bisimulation. $\mathcal{E}_1,\mathcal{E}_2$ are weakly hereditary history-preserving (hhp-)bisimilar and are written $\mathcal{E}_1\approx_{hhp}\mathcal{E}_2$.
\end{definition}

\begin{proposition}[Weakly concurrent behavioral equivalence]\label{WSCBE}
(Strongly) concurrent behavioral equivalences imply weakly concurrent behavioral equivalences. That is, $\sim_p$ implies $\approx_p$, $\sim_s$ implies $\approx_s$, $\sim_{hp}$ implies $\approx_{hp}$, $\sim_{hhp}$ implies $\approx_{hhp}$.
\end{proposition}

\begin{proof}
From the definition of weak pomset transition, weak step transition, weakly posetal product and weakly concurrent behavioral equivalence, it is easy to see that $\xrightarrow{e}=\xrightarrow{\epsilon}\xrightarrow{e}\xrightarrow{\epsilon}$ for $e\in \mathbb{E}$, where $\epsilon$ is the empty event.
\end{proof}

The weak transition rules for CTC are listed in Table \ref{WTRForCTC}.

\begin{center}
    \begin{table}
        \[\textbf{WAct}_1\quad \frac{}{\alpha.P\xRightarrow{\alpha}P}\]

        \[\textbf{WSum}_1\quad \frac{P\xRightarrow{\alpha}P'}{P+Q\xRightarrow{\alpha}P'}\]

        \[\textbf{WCom}_1\quad \frac{P\xRightarrow{\alpha}P'\quad Q\nrightarrow}{P\parallel Q\xRightarrow{\alpha}P'\parallel Q}\]

        \[\textbf{WCom}_2\quad \frac{Q\xRightarrow{\alpha}Q'\quad P\nrightarrow}{P\parallel Q\xRightarrow{\alpha}P\parallel Q'}\]

        \[\textbf{WCom}_3\quad \frac{P\xRightarrow{\alpha}P'\quad Q\xRightarrow{\beta}Q'}{P\parallel Q\xRightarrow{\{\alpha,\beta\}}P'\parallel Q'}\quad (\beta\neq\overline{\alpha})\]

        \[\textbf{WCom}_4\quad \frac{P\xRightarrow{l}P'\quad Q\xRightarrow{\overline{l}}Q'}{P\parallel Q\xRightarrow{\tau}P'\parallel Q'}\]

        \[\textbf{WAct}_2\quad \frac{}{(\alpha_1\parallel\cdots\parallel\alpha_n).P\xRightarrow{\{\alpha_1,\cdots,\alpha_n\}}P}\quad (\alpha_i\neq\overline{\alpha_j}\quad i,j\in\{1,\cdots,n\})\]

        \[\textbf{WSum}_2\quad \frac{P\xRightarrow{\{\alpha_1,\cdots,\alpha_n\}}P'}{P+Q\xRightarrow{\{\alpha_1,\cdots,\alpha_n\}}P'}\]

        \[\textbf{WRes}_1\quad \frac{P\xRightarrow{\alpha}P'}{P\setminus L\xRightarrow{\alpha}P'\setminus L}\quad (\alpha,\overline{\alpha}\notin L)\]

        \[\textbf{WRes}_2\quad \frac{P\xRightarrow{\{\alpha_1,\cdots,\alpha_n\}}P'}{P\setminus L\xRightarrow{\{\alpha_1,\cdots,\alpha_n\}}P'\setminus L}\quad (\alpha_1,\overline{\alpha_1},\cdots,\alpha_n,\overline{\alpha_n}\notin L)\]

        \[\textbf{WRel}_1\quad \frac{P\xRightarrow{\alpha}P'}{P[f]\xRightarrow{f(\alpha)}P'[f]}\]

        \[\textbf{WRel}_2\quad \frac{P\xRightarrow{\{\alpha_1,\cdots,\alpha_n\}}P'}{P[f]\xRightarrow{\{f(\alpha_1),\cdots,f(\alpha_n)\}}P'[f]}\]

        \[\textbf{WCon}_1\quad\frac{P\xRightarrow{\alpha}P'}{A\xRightarrow{\alpha}P'}\quad (A\overset{\text{def}}{=}P)\]

        \[\textbf{WCon}_2\quad\frac{P\xRightarrow{\{\alpha_1,\cdots,\alpha_n\}}P'}{A\xRightarrow{\{\alpha_1,\cdots,\alpha_n\}}P'}\quad (A\overset{\text{def}}{=}P)\]
        \caption{Weak transition rules of CTC}
        \label{WTRForCTC}
    \end{table}
\end{center}

\subsection{Laws and Congruence}

Remembering that $\tau$ can neither be restricted nor relabeled, by Proposition \ref{WSCBE}, we know that the monoid laws, the static laws and the new expansion law in section \ref{stcb} still hold with respect to the corresponding weakly truly concurrent bisimulations. And also, we can enjoy the full congruence of Prefix, Summation, Composition, Restriction, Relabelling and Constants with respect to corresponding weakly truly concurrent bisimulations. We will not retype these laws, and just give the $\tau$-specific laws.

\begin{proposition}[$\tau$ laws for weak step bisimulation]\label{TAUWSB}
The $\tau$ laws for weak step bisimulation is as follows.
\begin{enumerate}
  \item $P\approx_s \tau.P$;
  \item $\alpha.\tau.P\approx_s \alpha.P$;
  \item $(\alpha_1\parallel\cdots\parallel\alpha_n).\tau.P\approx_s (\alpha_1\parallel\cdots\parallel\alpha_n).P$;
  \item $P+\tau.P\approx_s \tau.P$;
  \item $\alpha.(P+\tau.Q)+\alpha.Q\approx_s\alpha.(P+\tau.Q)$;
  \item $(\alpha_1\parallel\cdots\parallel\alpha_n).(P+\tau.Q)+ (\alpha_1\parallel\cdots\parallel\alpha_n).Q\approx_s (\alpha_1\parallel\cdots\parallel\alpha_n).(P+\tau.Q)$;
  \item $P\approx_s \tau\parallel P$.
\end{enumerate}
\end{proposition}

\begin{proof}
Though transition rules in Table \ref{TRForCTC} are defined in the flavor of single event, they can be modified into a step (a set of events within which each event is pairwise concurrent), we omit them. If we treat a single event as a step containing just one event, the proof of $\tau$ laws does not exist any problem, so we use this way and still use the transition rules in Table \ref{TRForCTC}.

\begin{enumerate}
  \item $P\approx_s \tau.P$. By the weak transition rules $\textbf{WAct}_{1,2}$ of CTC in Table \ref{WTRForCTC}, we get

  $$\frac{P\xRightarrow{\alpha}P'}{P\xRightarrow{\alpha}P'}
  \quad \frac{P\xRightarrow{\alpha}P'}{\tau.P\xRightarrow{\alpha} P'}$$

  Since $P'\approx_s P'$, we get $P\approx_s \tau.P$, as desired.
  \item $\alpha.\tau.P\approx_s \alpha.P$. By the weak transition rules $\textbf{WAct}_{1,2}$ in Table \ref{WTRForCTC}, we get

  $$\frac{}{\alpha.\tau.P\xRightarrow{\alpha}P}
  \quad \frac{}{\alpha.P\xRightarrow{\alpha}P}$$

  Since $P\approx_s P$, we get $\alpha.\tau.P\approx_s \alpha.P$, as desired.
  \item $(\alpha_1\parallel\cdots\parallel\alpha_n).\tau.P\approx_s (\alpha_1\parallel\cdots\parallel\alpha_n).P$. By the weak transition rules $\textbf{WAct}_{1,2}$ in Table \ref{WTRForCTC}, we get

  $$\frac{}{(\alpha_1\parallel\cdots\parallel\alpha_n).\tau.P\xRightarrow{\{\alpha_1,\cdots,\alpha_n\}}P}
  \quad \frac{}{(\alpha_1\parallel\cdots\parallel\alpha_n).P\xRightarrow{\{\alpha_1,\cdots,\alpha_n\}}P}$$

  Since $P\approx_s P$, we get $(\alpha_1\parallel\cdots\parallel\alpha_n).\tau.P\approx_s (\alpha_1\parallel\cdots\parallel\alpha_n).P$, as desired.
  \item $P+\tau.P\approx_s \tau.P$. By the weak transition rules $\textbf{WSum}_{1,2}$ of CTC in Table \ref{WTRForCTC}, we get

  $$\frac{P\xRightarrow{\alpha}P'}{P+\tau.P\xRightarrow{\alpha}P'}
  \quad \frac{P\xRightarrow{\alpha}P'}{\tau.P\xRightarrow{\alpha} P'}$$

  Since $P'\approx_s P'$, we get $P+\tau.P\approx_s \tau.P$, as desired.
  \item $\alpha.(P+\tau.Q)+\alpha.Q\approx_s\alpha.(P+\tau.Q)$. By the weak transition rules $\textbf{WAct}_{1,2}$ and $\textbf{WSum}_{1,2}$ of CTC in Table \ref{WTRForCTC}, we get

  $$\frac{}{\alpha.(P+\tau.Q)+\alpha.Q\xRightarrow{\alpha}Q}
  \quad \frac{}{\alpha.(P+\tau.Q)\xRightarrow{\alpha} Q}$$

  Since $Q\approx_s Q$, we get $\alpha.(P+\tau.Q)+\alpha.Q\approx_s\alpha.(P+\tau.Q)$, as desired.
  \item $(\alpha_1\parallel\cdots\parallel\alpha_n).(P+\tau.Q)+ (\alpha_1\parallel\cdots\parallel\alpha_n).Q\approx_s (\alpha_1\parallel\cdots\parallel\alpha_n).(P+\tau.Q)$. By the weak transition rules $\textbf{WAct}_{1,2}$ and $\textbf{WSum}_{1,2}$ of CTC in Table \ref{WTRForCTC}, we get

  $$\frac{}{(\alpha_1\parallel\cdots\parallel\alpha_n).(P+\tau.Q)+(\alpha_1\parallel\cdots\parallel\alpha_n).Q \xRightarrow{\{\alpha_1,\cdots,\alpha_n\}}Q}
  \quad \frac{}{(\alpha_1\parallel\cdots\parallel\alpha_n).(P+\tau.Q)\xRightarrow{\{\alpha_1,\cdots,\alpha_n\}} Q}$$

  Since $Q\approx_s Q$, we get $(\alpha_1\parallel\cdots\parallel\alpha_n).(P+\tau.Q)+ (\alpha_1\parallel\cdots\parallel\alpha_n).Q\approx_s(\alpha_1\parallel\cdots\parallel\alpha_n).(P+\tau.Q)$, as desired.
  \item $P\approx_s \tau\parallel P$. By the weak transition rules $\textbf{WCom}_{1,2,3,4}$ of CTC in Table \ref{WTRForCTC}, we get

  $$\frac{P\xRightarrow{\alpha}P'}{P\xRightarrow{\alpha}P'}
  \quad \frac{P\xRightarrow{\alpha}P'}{\tau\parallel P\xRightarrow{\alpha} P'}$$

  Since $P'\approx_s P'$, we get $P\approx_s \tau\parallel P$, as desired.
\end{enumerate}
\end{proof}

\begin{proposition}[$\tau$ laws for weak pomset bisimulation]\label{TAUWPB}
The $\tau$ laws for weak pomset bisimulation is as follows.
\begin{enumerate}
  \item $P\approx_p \tau.P$;
  \item $\alpha.\tau.P\approx_p \alpha.P$;
  \item $(\alpha_1\parallel\cdots\parallel\alpha_n).\tau.P\approx_p (\alpha_1\parallel\cdots\parallel\alpha_n).P$;
  \item $P+\tau.P\approx_p \tau.P$;
  \item $\alpha.(P+\tau.Q)+\alpha.Q\approx_p\alpha.(P+\tau.Q)$;
  \item $(\alpha_1\parallel\cdots\parallel\alpha_n).(P+\tau.Q)+ (\alpha_1\parallel\cdots\parallel\alpha_n).Q\approx_p (\alpha_1\parallel\cdots\parallel\alpha_n).(P+\tau.Q)$;
  \item $P\approx_p \tau\parallel P$.
\end{enumerate}
\end{proposition}

\begin{proof}
From the definition of weak pomset bisimulation $\approx_{p}$ (see Definition \ref{WPSB}), we know that weak pomset bisimulation $\approx_{p}$ is defined by weak pomset transitions, which are labeled by pomsets with $\tau$. In a weak pomset transition, the events in the pomset are either within causality relations (defined by $.$) or in concurrency (implicitly defined by $.$ and $+$, and explicitly defined by $\parallel$), of course, they are pairwise consistent (without conflicts). In Proposition \ref{TAUWSB}, we have already proven the case that all events are pairwise concurrent, so, we only need to prove the case of events in causality. Without loss of generality, we take a pomset of $p=\{\alpha,\beta:\alpha.\beta\}$. Then the weak pomset transition labeled by the above $p$ is just composed of one single event transition labeled by $\alpha$ succeeded by another single event transition labeled by $\beta$, that is, $\xRightarrow{p}=\xRightarrow{\alpha}\xRightarrow{\beta}$.

Similarly to the proof of $\tau$ laws for weak step bisimulation $\approx_{s}$ (Proposition \ref{TAUWSB}), we can prove that $\tau$ laws hold for weak pomset bisimulation $\approx_{p}$, we omit them.
\end{proof}

\begin{proposition}[$\tau$ laws for weak hp-bisimulation]\label{TAUWHPB}
The $\tau$ laws for weak hp-bisimulation is as follows.
\begin{enumerate}
  \item $P\approx_{hp} \tau.P$;
  \item $\alpha.\tau.P\approx_{hp} \alpha.P$;
  \item $(\alpha_1\parallel\cdots\parallel\alpha_n).\tau.P\approx_{hp} (\alpha_1\parallel\cdots\parallel\alpha_n).P$;
  \item $P+\tau.P\approx_{hp} \tau.P$;
  \item $\alpha.(P+\tau.Q)+\alpha.Q\approx_{hp}\alpha.(P+\tau.Q)$;
  \item $(\alpha_1\parallel\cdots\parallel\alpha_n).(P+\tau.Q)+ (\alpha_1\parallel\cdots\parallel\alpha_n).Q\approx_{hp} (\alpha_1\parallel\cdots\parallel\alpha_n).(P+\tau.Q)$;
  \item $P\approx_{hp} \tau\parallel P$.
\end{enumerate}
\end{proposition}

\begin{proof}
From the definition of weak hp-bisimulation $\approx_{hp}$ (see Definition \ref{WHHPB}), we know that weak hp-bisimulation $\approx_{hp}$ is defined on the weakly posetal product $(C_1,f,C_2),f:\hat{C_1}\rightarrow \hat{C_2}\textrm{ isomorphism}$. Two processes $P$ related to $C_1$ and $Q$ related to $C_2$, and $f:\hat{C_1}\rightarrow \hat{C_2}\textrm{ isomorphism}$. Initially, $(C_1,f,C_2)=(\emptyset,\emptyset,\emptyset)$, and $(\emptyset,\emptyset,\emptyset)\in\approx_{hp}$. When $P\xrightarrow{\alpha}P'$ ($C_1\xrightarrow{\alpha}C_1'$), there will be $Q\xRightarrow{\alpha}Q'$ ($C_2\xRightarrow{\alpha}C_2'$), and we define $f'=f[\alpha\mapsto \alpha]$. Then, if $(C_1,f,C_2)\in\approx_{hp}$, then $(C_1',f',C_2')\in\approx_{hp}$.

Similarly to the proof of $\tau$ laws for weak pomset bisimulation (Proposition \ref{TAUWPB}), we can prove that $\tau$ laws hold for weak hp-bisimulation, we just need additionally to check the above conditions on weak hp-bisimulation, we omit them.
\end{proof}

\begin{proposition}[$\tau$ laws for weakly hhp-bisimulation]\label{TAUWHHPB}
The $\tau$ laws for weakly hhp-bisimulation is as follows.
\begin{enumerate}
  \item $P\approx_{hhp} \tau.P$;
  \item $\alpha.\tau.P\approx_{hhp} \alpha.P$;
  \item $(\alpha_1\parallel\cdots\parallel\alpha_n).\tau.P\approx_{hhp} (\alpha_1\parallel\cdots\parallel\alpha_n).P$;
  \item $P+\tau.P\approx_{hhp} \tau.P$;
  \item $\alpha.(P+\tau.Q)+\alpha.Q\approx_{hhp}\alpha.(P+\tau.Q)$;
  \item $(\alpha_1\parallel\cdots\parallel\alpha_n).(P+\tau.Q)+ (\alpha_1\parallel\cdots\parallel\alpha_n).Q\approx_{hhp} (\alpha_1\parallel\cdots\parallel\alpha_n).(P+\tau.Q)$;
  \item $P\approx_{hhp} \tau\parallel P$.
\end{enumerate}
\end{proposition}

\begin{proof}
From the definition of weakly hhp-bisimulation (see Definition \ref{WHHPB}), we know that weakly hhp-bisimulation is downward closed for weak hp-bisimulation.

Similarly to the proof of $\tau$ laws for weak hp-bisimulation (see Proposition \ref{TAUWHPB}), we can prove that the $\tau$ laws hold for weakly hhp-bisimulation, we omit them.
\end{proof}

\subsection{Recursion}

\begin{definition}[Sequential]
$X$ is sequential in $E$ if every subexpression of $E$ which contains $X$, apart from $X$ itself, is of the form $\alpha.F$, or $(\alpha_1\parallel\cdots\parallel\alpha_n).F$, or $\sum\widetilde{F}$.
\end{definition}

\begin{definition}[Guarded recursive expression]
$X$ is guarded in $E$ if each occurrence of $X$ is with some subexpression $l.F$ or $(l_1\parallel\cdots\parallel l_n).F$ of $E$.
\end{definition}

\begin{lemma}\label{LUSWW}
Let $G$ be guarded and sequential, $Vars(G)\subseteq\widetilde{X}$, and let $G\{\widetilde{P}/\widetilde{X}\}\xrightarrow{\{\alpha_1,\cdots,\alpha_n\}}P'$. Then there is an expression $H$ such that $G\xrightarrow{\{\alpha_1,\cdots,\alpha_n\}}H$, $P'\equiv H\{\widetilde{P}/\widetilde{X}\}$, and for any $\widetilde{Q}$, $G\{\widetilde{Q}/\widetilde{X}\}\xrightarrow{\{\alpha_1,\cdots,\alpha_n\}} H\{\widetilde{Q}/\widetilde{X}\}$. Moreover $H$ is sequential, $Vars(H)\subseteq\widetilde{X}$, and if $\alpha_1=\cdots=\alpha_n=\tau$, then $H$ is also guarded.
\end{lemma}

\begin{proof}
We need to induct on the structure of $G$.

If $G$ is a Constant, a Composition, a Restriction or a Relabeling then it contains no variables, since $G$ is sequential and guarded, then $G\xrightarrow{\{\alpha_1,\cdots,\alpha_n\}}P'$, then let $H\equiv P'$, as desired.

$G$ cannot be a variable, since it is guarded.

If $G\equiv G_1+G_2$. Then either $G_1\{\widetilde{P}/\widetilde{X}\} \xrightarrow{\{\alpha_1,\cdots,\alpha_n\}}P'$ or $G_2\{\widetilde{P}/\widetilde{X}\} \xrightarrow{\{\alpha_1,\cdots,\alpha_n\}}P'$, then, we can apply this lemma in either case, as desired.

If $G\equiv\beta.H$. Then we must have $\alpha=\beta$, and $P'\equiv H\{\widetilde{P}/\widetilde{X}\}$, and $G\{\widetilde{Q}/\widetilde{X}\}\equiv \beta.H\{\widetilde{Q}/\widetilde{X}\} \xrightarrow{\beta}H\{\widetilde{Q}/\widetilde{X}\}$, then, let $G'$ be $H$, as desired.

If $G\equiv(\beta_1\parallel\cdots\parallel\beta_n).H$. Then we must have $\alpha_i=\beta_i$ for $1\leq i\leq n$, and $P'\equiv H\{\widetilde{P}/\widetilde{X}\}$, and $G\{\widetilde{Q}/\widetilde{X}\}\equiv (\beta_1\parallel\cdots\parallel\beta_n).H\{\widetilde{Q}/\widetilde{X}\} \xrightarrow{\{\beta_1,\cdots,\beta_n\}}H\{\widetilde{Q}/\widetilde{X}\}$, then, let $G'$ be $H$, as desired.

If $G\equiv\tau.H$. Then we must have $\tau=\tau$, and $P'\equiv H\{\widetilde{P}/\widetilde{X}\}$, and $G\{\widetilde{Q}/\widetilde{X}\}\equiv \tau.H\{\widetilde{Q}/\widetilde{X}\} \xrightarrow{\tau}H\{\widetilde{Q}/\widetilde{X}\}$, then, let $G'$ be $H$, as desired.
\end{proof}

\begin{theorem}[Unique solution of equations for weak step bisimulation]\label{USWSB}
Let the guarded and sequential expressions $\widetilde{E}$ contain free variables $\subseteq \widetilde{X}$, then,

If $\widetilde{P}\approx_s \widetilde{E}\{\widetilde{P}/\widetilde{X}\}$ and $\widetilde{Q}\approx_s \widetilde{E}\{\widetilde{Q}/\widetilde{X}\}$, then $\widetilde{P}\approx_s \widetilde{Q}$.
\end{theorem}

\begin{proof}
Like the corresponding theorem in CCS, without loss of generality, we only consider a single equation $X=E$. So we assume $P\approx_s E(P)$, $Q\approx_s E(Q)$, then $P\approx_s Q$.

We will prove $\{(H(P),H(Q)): H\}$ sequential, if $H(P)\xrightarrow{\{\alpha_1,\cdots,\alpha_n\}}P'$, then, for some $Q'$, $H(Q)\xRightarrow{\{\alpha_1.\cdots,\alpha_n\}}Q'$ and $P'\approx_s Q'$.

Let $H(P)\xrightarrow{\{\alpha_1,\cdot,\alpha_n\}}P'$, then $H(E(P))\xRightarrow{\{\alpha_1,\cdots,\alpha_n\}}P''$ and $P'\approx_s P''$.

By Lemma \ref{LUSWW}, we know there is a sequential $H'$ such that $H(E(P))\xRightarrow{\{\alpha_1,\cdots,\alpha_n\}}H'(P)\Rightarrow P''\approx_s P'$.

And, $H(E(Q))\xRightarrow{\{\alpha_1,\cdots,\alpha_n\}}H'(Q)\Rightarrow Q''$ and $P''\approx_s Q''$. And $H(Q)\xrightarrow{\{\alpha_1,\cdots,\alpha_n\}}Q'\approx_s Q''$. Hence, $P'\approx_s Q'$, as desired.
\end{proof}

\begin{theorem}[Unique solution of equations for weak pomset bisimulation]\label{USWPB}
Let the guarded and sequential expressions $\widetilde{E}$ contain free variables $\subseteq \widetilde{X}$, then,

If $\widetilde{P}\approx_p \widetilde{E}\{\widetilde{P}/\widetilde{X}\}$ and $\widetilde{Q}\approx_p \widetilde{E}\{\widetilde{Q}/\widetilde{X}\}$, then $\widetilde{P}\approx_p \widetilde{Q}$.
\end{theorem}

\begin{proof}
From the definition of weak pomset bisimulation $\approx_{p}$ (see Definition \ref{WPSB}), we know that weak pomset bisimulation $\approx_{p}$ is defined by weak pomset transitions, which are labeled by pomsets with $\tau$. In a weak pomset transition, the events in the pomset are either within causality relations (defined by $.$) or in concurrency (implicitly defined by $.$ and $+$, and explicitly defined by $\parallel$), of course, they are pairwise consistent (without conflicts). In Theorem \ref{USWSB}, we have already proven the case that all events are pairwise concurrent, so, we only need to prove the case of events in causality. Without loss of generality, we take a pomset of $p=\{\alpha,\beta:\alpha.\beta\}$. Then the weak pomset transition labeled by the above $p$ is just composed of one single event transition labeled by $\alpha$ succeeded by another single event transition labeled by $\beta$, that is, $\xRightarrow{p}=\xRightarrow{\alpha}\xRightarrow{\beta}$.

Similarly to the proof of unique solution of equations for weak step bisimulation $\approx_{s}$ (Theorem \ref{USWSB}), we can prove that unique solution of equations holds for weak pomset bisimulation $\approx_{p}$, we omit them.
\end{proof}

\begin{theorem}[Unique solution of equations for weak hp-bisimulation]\label{USWHPB}
Let the guarded and sequential expressions $\widetilde{E}$ contain free variables $\subseteq \widetilde{X}$, then,

If $\widetilde{P}\approx_{hp} \widetilde{E}\{\widetilde{P}/\widetilde{X}\}$ and $\widetilde{Q}\approx_{hp} \widetilde{E}\{\widetilde{Q}/\widetilde{X}\}$, then $\widetilde{P}\approx_{hp} \widetilde{Q}$.
\end{theorem}

\begin{proof}
From the definition of weak hp-bisimulation $\approx_{hp}$ (see Definition \ref{WHHPB}), we know that weak hp-bisimulation $\approx_{hp}$ is defined on the weakly posetal product $(C_1,f,C_2),f:\hat{C_1}\rightarrow \hat{C_2}\textrm{ isomorphism}$. Two processes $P$ related to $C_1$ and $Q$ related to $C_2$, and $f:\hat{C_1}\rightarrow \hat{C_2}\textrm{ isomorphism}$. Initially, $(C_1,f,C_2)=(\emptyset,\emptyset,\emptyset)$, and $(\emptyset,\emptyset,\emptyset)\in\approx_{hp}$. When $P\xrightarrow{\alpha}P'$ ($C_1\xrightarrow{\alpha}C_1'$), there will be $Q\xRightarrow{\alpha}Q'$ ($C_2\xRightarrow{\alpha}C_2'$), and we define $f'=f[\alpha\mapsto \alpha]$. Then, if $(C_1,f,C_2)\in\approx_{hp}$, then $(C_1',f',C_2')\in\approx_{hp}$.

Similarly to the proof of unique solution of equations for weak pomset bisimulation (Theorem \ref{USWPB}), we can prove that unique solution of equations holds for weak hp-bisimulation, we just need additionally to check the above conditions on weak hp-bisimulation, we omit them.
\end{proof}

\begin{theorem}[Unique solution of equations for weakly hhp-bisimulation]\label{USWHHPB}
Let the guarded and sequential expressions $\widetilde{E}$ contain free variables $\subseteq \widetilde{X}$, then,

If $\widetilde{P}\approx_{hhp} \widetilde{E}\{\widetilde{P}/\widetilde{X}\}$ and $\widetilde{Q}\approx_{hhp} \widetilde{E}\{\widetilde{Q}/\widetilde{X}\}$, then $\widetilde{P}\approx_{hhp} \widetilde{Q}$.
\end{theorem}

\begin{proof}
From the definition of weakly hhp-bisimulation (see Definition \ref{WHHPB}), we know that weakly hhp-bisimulation is downward closed for weak hp-bisimulation.

Similarly to the proof of unique solution of equations for weak hp-bisimulation (see Theorem \ref{USWHPB}), we can prove that the unique solution of equations holds for weakly hhp-bisimulation, we omit them.
\end{proof}

\section{Applications}\label{app}

In this section, we show the applications of CTC by verification of the alternating-bit protocol \cite{ABP}. The alternating-bit protocol is a communication protocol and illustrated in Fig. \ref{ABP}.

\begin{figure}
    \centering
    \includegraphics{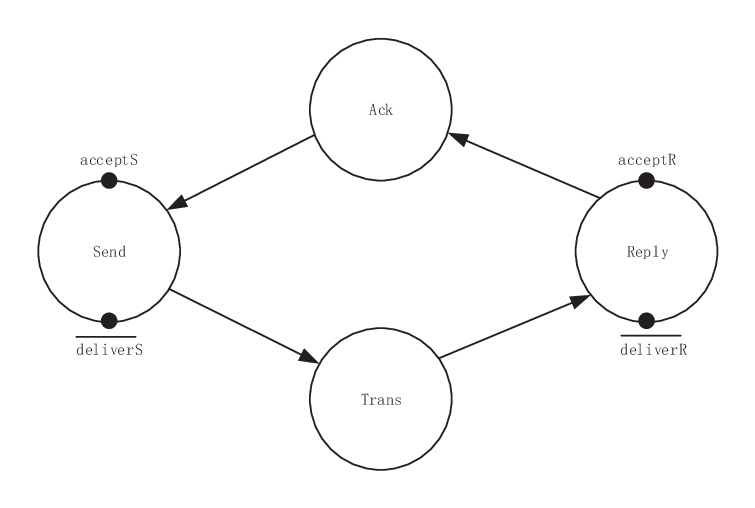}
    \caption{Alternating bit protocol}
    \label{ABP}
\end{figure}

The $Trans$ and $Ack$ lines may lose or duplicate message, and messages are sent tagged with the bits $0$ or $1$ alternately, and these bits also constitute the acknowledgements.

There are some variations of the classical alternating-bit protocol. We assume the message flows are bidirectional, the sender accepts messages from the outside world, and sends it to the replier via some line; and it also accepts messages from the replier, deliver it to the outside world and acknowledges it via some line. The role of the replier acts the same as the sender. But, for simplicities and without loss the generality, we only consider the dual one directional processes, we just suppose that the messages from the replier are always accompanied with the acknowledgements.

After accepting a message from the outside world, the sender sends it with bit $b$ along the $Trans$ line and sets a timer,

\begin{enumerate}
  \item it may get a "time-out" from the timer, upon which it sends the message again with $b$;
  \item it may get an acknowledgement $b$ accompanied with the message from the replier via the $Ack$ line, upon which it deliver the message to the outside world and is ready to accept another message tagged with bit $\hat{b}=1-b$;
  \item it may get an acknowledgement $\hat{b}$ which it ignores.
\end{enumerate}

After accepting a message from outside world, the replier also can send it to the sender, but we ignore this process and just assume the message is always accompanied with the acknowledgement to the sender, and just process the dual manner to the sender. After delivering a message to the outside world it acknowledges it with bit $b$ along the $Ack$ line and sets a timer,

\begin{enumerate}
  \item it may get a "time-out" from the timer, upon which it acknowledges again with $b$;
  \item it may get a new message with bit $\hat{b}$ from the $Trans$ line, upon which it is ready to deliver the new message and acknowledge with bit $\hat{b}$ accompanying with its messages from the outside world;
  \item it may get a superfluous transmission of the previous message with bit $b$, which it ignores.
\end{enumerate}

Now, we give the formal definitions as follows.

\[Send(b)\overset{\text{def}}{=}\overline{send_b}.\overline{time}.Sending(b)\]

\[Sending(b)\overset{\text{def}}{=}timeout.Send(b)+ack_b.timeout.AcceptS(\hat{b}) +ack_{\hat{b}}.DeliverS(b)\]

\[AcceptS(b)\overset{\text{def}}{=}acceptS.Send(b)\]

\[DeliverS(b)\overset{\text{def}}{=}\overline{deliverS}.Sending(b)\]

\[Reply(b)\overset{\text{def}}{=}\overline{reply_b}.\overline{time}.Replying(b)\]

\[Replying(b)\overset{\text{def}}{=}timeout.Reply(b)+trans_{\hat{b}}.timeout.Deliver(\hat{b}) +trans_b.AcceptR(b)\]

\[DeliverR(b)\overset{\text{def}}{=}\overline{deliverR}.Reply(b)\]

\[AcceptR(b)\overset{\text{def}}{=}acceptR.Replying(b)\]

\[Timer\overset{\text{def}}{=}time.\overline{timeout}.Timer\]

\[Ack(bs)\xrightarrow{\overline{ack_b}}Ack(s)\quad Trans(sb)\xrightarrow{\overline{trans_b}}Trans(s)\]

\[Ack(s)\xrightarrow{reply_b}Ack(sb)\quad Trans(s)\xrightarrow{send_b}Trans(bs)\]

\[Ack(sbt)\xrightarrow{\tau}Ack(st)\quad Trans(tbs)\xrightarrow{\tau}Trans(ts)\]

\[Ack(sbt)\xrightarrow{\tau}Ack(sbbt)\quad Trans(tbs)\xrightarrow{\tau}Trans(tbbs)\]

Then the complete system can be builded by compose the components. That is, it can be expressed as follows, where $\epsilon$ is the empty sequence.

$$AB\overset{\text{def}}{=}Accept(\hat{b})\parallel Trans(\epsilon)\parallel Ack(\epsilon)\parallel Reply(b)\parallel Timer$$

Now, we define the protocol specification to be a buffer as follows:

$$Buff\overset{\text{def}}{=}(acceptS\parallel acceptR).Buff'$$

$$Buff'\overset{\text{def}}{=}(\overline{deliverS}\parallel \overline{deliverR}).Buff$$

We need to prove that

$$AB\approx_s Buff$$

$$AB\approx_p Buff$$

$$AB\approx_{hp} Buff$$

$$AB\approx_{hhp} Buff$$

The deductive process is omitted, and we left it as an excise to the readers. 

\section{Conclusions}\label{con}

We design a calculus for true concurrency (CTC). Indeed, we follow the way paved by Milner's famous CCS \cite{CCS} \cite{CC} for interleaving bisimulation.

Fortunately, based on the concepts for true concurrency, CTC has good properties modulo several kinds of strongly truly concurrent bisimulations and weakly truly concurrent bisimulations. These properties include monoid laws, static laws, new expansion law for strongly truly concurrent bisimulations, $\tau$ laws for weakly truly concurrent bisimulations, and full congruences for strongly and weakly truly concurrent bisimulations, and also unique solution for recursion.

CTC is a peer in true concurrency to CCS in interleaving bisimulation semantics. It can be used widely in verification of computer systems with a truly concurrent flavor.

\label{lastpage}

\end{document}